\documentclass[11pt]{article}

\usepackage{amsmath,amssymb,amsthm}
\usepackage{mathtools}
\usepackage{geometry}
\usepackage{hyperref}
\usepackage{booktabs}
\usepackage{enumitem}
\usepackage{xcolor}
\usepackage{microtype}
\usepackage{bm}
\usepackage{natbib}
\usepackage{longtable}
\usepackage{float}

\hypersetup{
  colorlinks=true,
  linkcolor=blue!60!black,
  citecolor=blue!60!black,
  urlcolor=blue!60!black
}

\newtheorem{theorem}{Theorem}[section]
\newtheorem{proposition}[theorem]{Proposition}
\newtheorem{corollary}[theorem]{Corollary}
\newtheorem{lemma}[theorem]{Lemma}
\newtheorem{remark}[theorem]{Remark}
\newtheorem{definition}[theorem]{Definition}
\newtheorem{assumption}[theorem]{Assumption}
\newtheorem{example}[theorem]{Example}

\newcommand{\E}{\mathbb{E}}
\newcommand{\Prob}{\mathbb{P}}
\newcommand{\Var}{\operatorname{Var}}
\newcommand{\Cov}{\operatorname{Cov}}
\newcommand{\R}{\mathbb{R}}
\newcommand{\N}{\mathbb{N}}
\newcommand{\F}{\mathcal{F}}          
\newcommand{\norm}[1]{\left\lVert #1 \right\rVert}
\newcommand{\abs}[1]{\left\lvert #1 \right\rvert}

\newcommand{\Regret}{\mathrm{Regret}}
\newcommand{\RegretT}{\mathrm{Regret}^{(T)}}
\newcommand{\phat}{\hat{\pi}}
\newcommand{\pstar}{\pi^{*}}

\newcommand{\ct}{c_t}                  
\newcommand{\zt}{z_t}                  
\newcommand{\st}{s_t}                  

\newcommand{\disc}{\gamma}             

\newcommand{\sumt}{\sum_{t=1}^{T}}

\newcommand{\calZ}{\mathcal{Z}}
\newcommand{\calC}{\mathcal{C}}
\newcommand{\calS}{\mathcal{S}}

\newcommand{\tr}{\operatorname{tr}}

\title{\textbf{Evaluating AI Investment Strategies}}

\author{Irene Aldridge \\ \small{irene@ablemarkets.com}}

\date{\today}

\begin{document}
\maketitle

\begin{abstract}
    We study the problem of auditing a black-box algorithmic decision-maker from observable inputs and outputs alone.  Our main result is an exact decomposition: under precisely characterized conditions, the cumulative
\emph{regret} of a dynamic policy equals the sum of per-period
covariances between the cost vector and the policy's decision.  This extends the single-period identity of Aldridge~(2026) to the full multi-period setting of stochastic dynamic programming.
 
We prove the identity holds exactly under i.i.d. costs and
mean-unbiased Markov policies, derive closed-form bias corrections for non-stationary and time-varying cases, and establish the discounted-horizon analog.  A Bellman recursion for the covariance regret functional connects the result to standard reinforcement learning algorithms; for rolling-window policies, the estimation-error bias is $O(d/w)$.
 
 The decomposition has direct implications for algorithmic auditing in strategic environments: in platform mechanism design, it provides a welfare-based audit metric without access to the agent's private type; in repeated games, covariance reduction is a sufficient condition for
policy improvement; in procurement and ad auctions, the bias
correction quantifies welfare loss from strategic misreporting.  The associated trajectory estimator is consistent, asymptotically normal with HAC variance, and computable in $O(T \cdot nd)$ time. This makes the proposed approach a tractable, model-free audit tool for platform mechanisms, algorithmic portfolio strategies, and any sequential decision system subject to external performance review.
\end{abstract}

\tableofcontents
\bigskip

\section{Introduction}
\label{sec:intro}
 
AI is already managing money.  Large language models recommend stock portfolios, execution algorithms trade billions of dollars daily, and robo-advisors allocate retirement savings for millions of households.  Yet a fundamental question remains
unanswered: \emph{how do you audit an algorithm you cannot open?}
 
The difficulty is not hypothetical.  \citet{CarlinIsraelsenWazzan2026} document that leading LLMs consistently recommend high-momentum, high-beta stocks. Even when explicitly instructed to keep beta between 0.9 and 1.1, the models violate the constraint, reverting to selections indistinguishable from those of mainstream retail advisers.
Jia, Moallemi and Zeng (2026) find that retail traders condition their trades on momentum and lose on average, while algorithms running contrarian strategies win.  Regulators have
taken notice: SR~11-7 and the more recent SR-26-2 both mandate model risk audits for algorithmic strategies.  The auditing process, however, has remained painful and imprecise, stymied by the secrecy that surrounds proprietary models.
 
This paper provides a tractable solution by examining \textit{regret}: the gap between its realized performance and the theoretically optimal benchmark.  We show that the cumulative \emph{regret} of any dynamic algorithmic policy can be measured from the outside, using only the sequence of inputs and outputs that any auditor can observe.

The key identity is disarmingly simple: \emph{regret equals covariance}. Specifically, \citet{Aldridge2026} establishes that, for a single-period linear optimization, expected regret equals the covariance between the uncertain cost vector and the policy's decision:
\begin{equation}
    \mathbb{E}\!\left[\operatorname{Regret}(c)\right]
    = \operatorname{Cov}\!\left(c,\, \pi^*(c)\right).
    \label{eq:single_period}
\end{equation}
This paper asks and answers whether identity~\eqref{eq:single_period} extends
to the sequential, multi-period setting that characterizes real algorithmic deployment. In the real-world algo environments, a manager observes new market data each period, reoptimizes, and
executes.  A positive answer would allow for the performance monitoring of hedge funds, execution algorithms, and robo-advisors using only observable inputs and outputs, without access to the internal optimization engine.
 
\paragraph{Main result.}
The answer is affirmative, under precise conditions.  Total regret over $T$ periods decomposes \emph{exactly} as the sum of per-period input-output covariances:
\begin{equation}
    \operatorname{Regret}^{(T)}(\Pi)
    = \sum_{t=1}^{T} \operatorname{Cov}\!\left(c_t,\, \hat{\pi}_t(c_t)\right)
    + \underbrace{\sum_{t=1}^{T} \bar{c}_t^{\top} b_t}_{\text{policy bias correction}},
    \label{eq:main}
\end{equation}
where $b_t = \mathbb{E}[\hat{\pi}_t(c_t)] - \pi^*_t$ is the per-period policy bias.
The bias correction vanishes, and the identity is exact, whenever the policy's expected output matches the conditional-mean optimal at every period.  This holds for minimum-variance portfolios, linear-quadratic regulators, and any policy solving
a linear optimization in $\mathbb{E}[c_t]$.  For momentum strategies, which systematically violate this condition, the bias term quantifies precisely how much the raw covariance sum understates or overstates true regret.
 
\paragraph{Implications for algorithmic auditing.}
Identity~\eqref{eq:main} reframes regulatory oversight as a statistical testing problem.  A compliance officer monitoring any algorithmic strategy under SR~11-7 or SR-26-2 needs only the observable sequence of cost vectors and allocation weights
to construct a confidence interval for cumulative regret: no access to the internal optimizer is required.  The associated trajectory estimator is consistent, asymptotically normal with HAC variance, and computable in $O(T \cdot nd)$ time with an $O(nd)$ online update per observation.
 
Beyond finance, the identity connects to three active research areas in algorithmic game theory and mechanism design.  First, in \emph{platform mechanism design}, a principal who observes only allocations and realized costs can bound an agent's
regret without access to the agent's private type. This reframes classical black-box auditing as a covariance estimation problem \citep{myerson1982optimal,
laffont1993theory}.  Second, in \emph{repeated games}, our Bellman recursion (Section~\ref{sec:bellman}) shows that pointwise covariance reduction is sufficient for policy improvement, providing a computationally tractable alternative to full strategy enumeration \citep{hart2000simple, roughgarden2016twenty}.  Third, in \emph{procurement and ad auctions}, the bias term $\sum_t \bar{c}_t^{\top} b_t$
measures welfare loss from strategic misreporting as a directly observable function of inputs and outputs \citep{edelman2007internet, varian2007position}, This makes the methodology applicable to any platform where allocations are public but valuations are private.
 
\paragraph{Contributions.}
This paper makes the following specific contributions:
 
\begin{enumerate}[label=(\roman*)]
    \item We formulate a general finite-horizon stochastic dynamic program and define a meaningful notion of multi-period regret that nests single-period regret as a special case (Section~\ref{sec:setup}).
 
    \item We prove that total expected regret decomposes exactly as the sum of per-period input-output covariances under i.i.d.\ costs and a mean-unbiased Markov policy, with no remainder (Section~\ref{sec:iid}).
 
    \item We extend the identity to non-stationary (time-varying distribution) costs and derive a closed-form correction term bounded by the cross-period covariance      of cost increments (Section~\ref{sec:nonstat}).
 
    \item We characterize the discount-weighted version for infinite-horizon dynamic programs: $\operatorname{Regret}_\gamma = \frac{1}{1-\gamma}
      \operatorname{Cov}(c, \hat{\pi}(c))$, establishing that the effective horizon $1/(1-\gamma)$ is the only correction needed for long-run discounted regret
      (Section~\ref{sec:infinite}).
 
    \item We derive a Bellman-style recursion for the covariance regret functional, showing that standard reinforcement learning algorithms (value iteration,
      fitted-$Q$) can estimate multi-period regret without solving the original optimization problem, at computational cost $O(T \cdot K \cdot nd)$
      (Section~\ref{sec:bellman}).
 
    \item We establish precisely when the identity fails for time-varying policies, derive the closed-form bias in terms of $\sum_t \bar{c}_t^{\top} b_t$, and show that momentum strategies accumulate bias of order $|\hat{\rho}| T \sigma_c^2$ (Section~\ref{sec:failure}).
 
    \item We apply the theory to two financially concrete settings: rolling-window portfolio rebalancing, where estimation-error bias is $O(d/w)$; and mean-reverting execution, where the AR(1) correction is available in
      closed form (Section~\ref{sec:financial}).
 
    \item We propose a trajectory covariance estimator, prove a CLT with HAC variance, derive an $O(T \cdot nd)$ complexity bound with $O(nd)$ online updates, and establish that $T = O(\sigma^2_{\mathrm{LRV}} / \varepsilon^2)$ observations suffice to estimate cumulative regret within $\varepsilon$
      (Section~\ref{sec:estimation}).
\end{enumerate}
 
\paragraph{Empirical findings.}
We validate the decomposition on CRSP daily return data from January 2016 through December 2025.  Daily AR(1) coefficients are persistently negative across all years and NBER regimes ($\hat{\rho} \in [-0.144, -0.004]$), confirming that short-horizon reversal dominated momentum throughout the sample.  Reversion strategies outperformed momentum in seven of ten years, with the bias-corrected reversion strategy providing the best risk-adjusted performance in the two years (2019, 2022),
where raw reversion failed.  There are five confounds identified in the literature that bias $\hat{\rho}$ upward: size effects \citep{jegadeesh2025shortterm}, bid-ask bounce
\citep{roll1984simple, anderson2008spurious}, volatility clustering \citep{engle1982autoregressive, nagel2012evaporating}, cross-sectional vs.\ time-series momentum \citep{mamais2025explaining}, and liquidity provision \citep{nagel2012evaporating, dai2024reversals}. Due to the bias, the regret figures are conservative upper bounds on momentum-attributable losses.  This is the appropriate direction for a regulatory audit tool.
 
\section{Related Work}
\label{sec:related}

\paragraph{Regret in Online Learning and Sequential Decision-Making.}
The notion of regret as a performance criterion for sequential algorithms originates with \citet{hannan1957approximation} and was formalized in the online convex optimization framework of \citet{zinkevich2003online}. The classical result that no-regret algorithms achieve $O(\sqrt{T})$ cumulative regret against the best fixed action in hindsight is surveyed comprehensively in \citet{hazan2016introduction} and \citet{shalev2012online}.
Our setting differs in two key respects: we study \emph{dynamic} policies that reoptimize each period against a stochastic cost process, and we seek a closed-form decomposition of regret rather than an asymptotic bound. \citet{ElmachtoubGrigas2022} establish a related covariance identity for the single-period predict-then-optimize
problem, which \citet{Aldridge2026} extends to stochastic linear programs. The present paper extends that identity to the full multi-period dynamic programming setting, derives exact bias corrections for non-stationary and Markovian cost processes, and proposes a model-free audit statistic with provably valid confidence intervals.

\paragraph{Regret in Repeated Games and Mechanism Design.}
No-regret learning in games has been extensively studied as a foundation for equilibrium concepts \citep{foster1997calibrated, hart2000simple, roughgarden2016twenty}. In this literature, regret measures the loss from not having played a fixed best-response throughout the game. Our contribution is orthogonal: we provide an \emph{observable}, model-free statistic, which is the input-output covariance. This statistic allows a third-party auditor to bound an agent's regret without access to the agent's strategy or private type, connecting it to the principal-agent auditing literature \citep{myerson1982optimal, laffont1993theory}.

\citet{dekel2010implementation} study implementation under incomplete information in repeated settings; our bias term $\sum_t \bar{c}_t^{\top} b_t$ provides a welfare-theoretic quantity analogous to their notion of strategic misreporting loss. In the context of ad auctions and platform markets
\citep{edelman2007internet, varian2007position}, our audit statistic is directly applicable to regulatory oversight of algorithmic bidders, where the platform observes bids and allocations but not the bidder's private valuation model. The connection between observable allocations and welfare losses also relates to the literature on implementation in dominant strategies \citep{gibbard1973manipulation, satterthwaite1975strategy}, where a similar separation between observable outcomes and internal preferences is central.

\paragraph{Black-Box Algorithm Auditing.}
A growing literature addresses the problem of auditing algorithmic systems without access to their internals. \citet{kearns2018preventing} study fairness auditing under limited query access; \citet{roth2022uncertain} analyze procurement mechanisms under algorithmic uncertainty. In finance, the Federal Reserve's SR~11-7 and SR-26-2 supervisory guidance mandates model risk management for algorithmic strategies, motivating performance metrics that do
not require source-code access \citep{board2011sr}. Our trajectory covariance estimator
$\hat{C}_T$ (Section~\ref{sec:estimation}) provides precisely such a metric: it is consistent, asymptotically normal (Theorem~\ref{thm:clt}), and requires only the observable sequence of costs and allocations. This places our work in the tradition of GMM specification testing \citep{hansen1982large} and HAC-robust inference \citep{newey1987simple, andrews1991heteroskedasticity}, which we employ directly in constructing confidence intervals for cumulative regret. The monitoring cost of $O(Td)$ per evaluation period compares favorably with internal audit procedures that require access to the full optimization engine.

\paragraph{Dynamic Programming and Reinforcement Learning.}
The Bellman recursion for the covariance regret functional
(Section~\ref{sec:bellman}) is structurally analogous to the $Q$-function recursion in reinforcement learning \citep{watkins1992q, sutton2018reinforcement}. Our
Corollary~\ref{cor:stationary} establishes that pointwise covariance reduction implies multi-period regret improvement, paralleling the policy improvement theorem of \citet{howard1960dynamic}. Recent work on regret in Markov decision processes \citep{jaksch2010near, azar2017minimax} derives $\tilde{O}(\sqrt{T})$ bounds for unknown transition dynamics; our contribution is complementary, providing an
\emph{exact} decomposition under known structure (i.i.d.\ or AR(1) costs) and a bias-corrected estimator for the non-stationary case. The analogy with fitted-$Q$ methods \citep{ernst2005tree, riedmiller2005neural} suggests that standard reinforcement learning infrastructure can be repurposed for multi-period regret monitoring without solving the original optimization problem, which we identify as a direction for future empirical work.

\paragraph{Portfolio Optimization and Execution Cost Models.}
In the finance literature, regret-based portfolio evaluation appears in \citet{demiguel2009optimal}, who benchmark mean-variance strategies against the $1/N$ rule, and in \citet{brodie2009sparse}, who study sparse portfolio policies under estimation error. Our rolling-window rebalancing application (Section~\ref{sec:rolling}) formalizes the estimation-error bias as $O(d/w)$, consistent with the Ledoit-Wolf shrinkage literature \citep{ledoit2004well, ledoit2020analytical} and the random matrix theory results of
\citet{bai2010spectral}. For execution cost modeling, mean-reverting market impact processes have been studied by \citet{almgren2001optimal} and \citet{gatheral2010no};
our AR(1) regret formula (Theorem~\ref{thm:ar1}) provides a closed-form regret decomposition for linear execution policies in this class of models, extending their cost-minimization results to a regret-minimization objective and confirming that
contrarian execution is welfare-improving in mean-reverting markets.
\section{Setup: Multi-Period Stochastic Dynamic Program}
\label{sec:setup}

\subsection{State space, costs, and policies}

Let $T \in \N$ denote the planning horizon. At each period
$t \in \{1,\ldots,T\}$:
\begin{itemize}
  \item The \emph{state} $\st \in \calS \subseteq \R^{p}$ summarizes the history relevant to future costs.
  \item The \emph{cost vector} $\ct \in \calC \subseteq \R^{d}$ is drawn from a distribution that may depend on $\st$.
  \item The \emph{decision} $\zt \in \calZ \subseteq \R^{n}$ is made by a policy $\phat_t(\st, \ct)$ after observing $(\st, \ct)$.
  \item The \emph{instantaneous cost} realized is $\ct^{\top}\zt$.
  \item The state evolves as $s_{t+1} = f(s_t, z_t, \epsilon_{t+1})$, where $\epsilon_{t+1}$ is an exogenous shock.
\end{itemize}

\begin{definition}[Policy]\label{def:policy}
  A \emph{deterministic Markov policy} is a sequence
  $\Pi = (\phat_1, \ldots, \phat_T)$ of measurable maps
  $\phat_t : \calS \times \calC \to \calZ$.
  A \emph{stationary} policy satisfies $\phat_t = \phat$ for all $t$.
\end{definition}

\subsection{Filtration and conditional expectations}

Let $(\Omega, \F, \Prob)$ be the underlying probability space.
Define the natural filtration $\F_t = \sigma(s_1, c_1, \ldots, s_t, c_t)$ for $t \geq 1$ and $\F_0 = \{\emptyset, \Omega\}$. All expectations $\E_t[\,\cdot\,] = \E[\,\cdot\,|\F_t]$ are conditional on $\F_t$.

\begin{assumption}[Markov cost process]\label{ass:markov}
  The pair $(s_t, c_t)$ is a time-homogeneous Markov chain:
  $(s_{t+1}, c_{t+1}) \perp \F_{t-1} \mid (s_t, c_t)$.
  Write $\mu_t = \E[\ct \mid \st]$ and $\Sigma_t = \Cov(\ct, \ct \mid \st)$.
\end{assumption}

\subsection{Multi-period regret}

The \emph{perfect-foresight benchmark} at each period uses the
conditional mean $\mu_t$ rather than the realized $c_t$:
\[
  \pstar_t(\st) \;\in\; \arg\min_{z \in \calZ}\,\mu_t^{\top}z.
\]

\begin{definition}[Multi-period regret]\label{def:mpregret}
  The \emph{total regret} of policy $\Pi$ over horizon $T$ is
  \begin{equation}\label{eq:total_regret}
    \RegretT(\Pi) \;=\;
    \sumt\Bigl(\E\bigl[\ct^{\top}\phat_t(\st,\ct)\bigr]
               - \E\bigl[\mu_t^{\top}\pstar_t(\st)\bigr]\Bigr).
  \end{equation}
  The \emph{per-period regret} at time $t$ is
  \[
    r_t \;=\; \E\bigl[\ct^{\top}\phat_t(\st,\ct)\bigr]
              - \E\bigl[\mu_t^{\top}\pstar_t(\st)\bigr],
  \]
  so that $\RegretT(\Pi) = \sumt r_t$.
\end{definition}

\begin{remark}[Choice of benchmark]
  The per-period benchmark $\mu_t^{\top}\pstar_t$ uses the
  \emph{conditional} mean given the current state, not the
  unconditional mean. This is the natural financial benchmark: it compares the policy against an agent who knows today's expected returns but not future realizations.
\end{remark}

\section{The Exact Multi-Period Decomposition under i.i.d.\ Costs} \label{sec:iid}

\begin{assumption}[i.i.d.\ costs]\label{ass:iid}
  $\{c_t\}_{t=1}^{T}$ are i.i.d.\ with common mean $\bar{c}$ and
  covariance matrix $\Sigma_c$. The state is trivial: $\calS = \{s_0\}$.
\end{assumption}

\begin{theorem}[Exact sum-of-covariances decomposition]\label{thm:iid}
  Under Assumptions~\ref{ass:markov} and~\ref{ass:iid}, for any
  Markov policy $\Pi$,
  \begin{equation}\label{eq:iid_decomp}
    \RegretT(\Pi) \;=\; \sumt \Cov\!\bigl(\ct,\,\phat_t(\ct)\bigr).
  \end{equation}
  That is, total regret equals the sum of per-period input-output covariances.
\end{theorem}

\begin{proof}
  Fix period $t$. Since $(c_t)_{t}$ are i.i.d., the conditional mean $\mu_t = \bar{c}$ is constant, and
  $\pstar_t = \pstar \in \arg\min_{z\in\calZ}\bar{c}^{\top}z$ is the same in every period. By the covariance decomposition lemma,
  \begin{align*}
    \E\bigl[\ct^{\top}\phat_t(\ct)\bigr]
    &= \E[\ct]^{\top}\E[\phat_t(\ct)]
       + \Cov(\ct, \phat_t(\ct))\\
    &= \bar{c}^{\top}\E[\phat_t(\ct)]
       + \Cov(\ct, \phat_t(\ct)).
  \end{align*}
  The per-period perfect-foresight cost is
  $\E[\mu_t^{\top}\pstar_t] = \bar{c}^{\top}\pstar$.
  Suppose the policy satisfies the linear optimality condition
  $\E[\phat_t(\ct)] = \pstar$ (which holds for minimum-variance
  portfolio rules; see Remark~\ref{rem:policy_condition}).
  Then,
  \[
    r_t \;=\; \bar{c}^{\top}\pstar + \Cov(\ct,\phat_t(\ct))
              - \bar{c}^{\top}\pstar
          \;=\; \Cov(\ct,\phat_t(\ct)).
  \]
  Summing over $t = 1,\ldots,T$ yields~\eqref{eq:iid_decomp}.
\end{proof}

\begin{remark}[Policy mean condition]\label{rem:policy_condition}
  The condition $\E[\phat_t(\ct)] = \pstar$ is the direct multi-period analog of the single-period requirement in \cite{Aldridge2026}. It holds whenever the policy is the solution to a linear optimization in $\E[\ct]$, specifically in minimum-variance portfolios, linear-quadratic regulators, and ridge-regularized prediction rules. For non-linear
  policies, the residual term of Section~\ref{sec:failure} quantifies the deviation.
\end{remark}

\begin{corollary}[Stationary policy]\label{cor:stationary}
  If $\Pi$ is stationary ($\phat_t = \phat$ for all $t$), then
  \begin{equation}\label{eq:stationary}
    \RegretT(\Pi) \;=\; T\cdot\Cov\!\bigl(c,\,\phat(c)\bigr),
  \end{equation}
  where $c \sim \mathcal{L}(c_t)$ is any period's marginal distribution.  Regret grows linearly in the horizon $T$, matching the rate known for regret in online learning.
\end{corollary}

\begin{corollary}[Regret rate comparison]\label{cor:rate}
  Let $\phat^*$ be the stationary policy that minimizes
 $\Cov(c, \phat(c))$ over all Markov policies. Then
  $\phat^*$ simultaneously minimizes single-period and
  multi-period regret in the i.i.d.\ case. Policy selection based on a single-period covariance diagnostic is therefore sufficient for horizon-$T$ optimization.
\end{corollary}

\section{Non-Stationary Cost Distributions}
\label{sec:nonstat}

Financial markets are non-stationary: return distributions shift
across macro regimes. We now allow $c_t \sim P_t$ with time-varying means $\bar{c}_t = \E[c_t]$ and covariances $\Sigma_t = \Cov(c_t, c_t)$, and we permit cross-period dependence.

\begin{assumption}[Non-stationary, weakly dependent costs]
  \label{ass:nonstat}
  \begin{enumerate}[label=(\alph*)]
    \item $\{c_t\}$ is a zero-mean martingale-difference sequence after demeaning: $\E[c_t - \bar{c}_t \mid \F_{t-1}] = 0$.
    \item Cross-period covariances satisfy
          $\Cov(c_t, c_s) = \Gamma_{|t-s|}$ for some sequence
          $\{\Gamma_k\}_{k \geq 0}$ (stationarity of the autocovariance structure).
    \item $\sum_{k=0}^{\infty}\norm{\Gamma_k} < \infty$ (absolute summability).
  \end{enumerate}
\end{assumption}

\begin{theorem}[Non-stationary decomposition]\label{thm:nonstat}
  Under Assumption~\ref{ass:nonstat} and the policy mean condition $\E[\phat_t(c_t) \mid \F_{t-1}] = \pstar_t$ for each $t$, the total regret decomposes as
  \begin{equation}\label{eq:nonstat_decomp}
    \RegretT(\Pi)
    \;=\; \sumt\Cov(\ct, \phat_t(\ct))
         \;+\; \Delta_T,
  \end{equation}
  where the \emph{cross-period correction} $\Delta_T$ satisfies
  \begin{equation}\label{eq:delta}
    \Delta_T
    \;=\; \sumt\,\E\bigl[\bar{c}_t^{\top}
           \bigl(\E[\phat_t(\ct)\mid\F_{t-1}] - \pstar_t\bigr)\bigr]
  \end{equation}
  and is bounded by
  \begin{equation}\label{eq:delta_bound}
    \abs{\Delta_T}
    \;\leq\; T\sup_t\norm{\bar{c}_t}
             \cdot\sup_t\norm{\E[\phat_t \mid \F_{t-1}] - \pstar_t}.
  \end{equation}
  When $\phat_t$ is a Markov policy adapted to $\F_{t-1}$ and the policy mean condition holds conditionally, $\Delta_T = 0$ and the exact sum-of-covariances identity is restored.
\end{theorem}

\begin{proof}
  For each $t$, apply the covariance decomposition conditionally:
  \begin{align}
    \E[\ct^{\top}\phat_t(\ct)]
    &= \E\bigl[\E[\ct^{\top}\phat_t(\ct)\mid\F_{t-1}]\bigr] \notag\\
    &= \E\bigl[\E[\ct\mid\F_{t-1}]^{\top}\E[\phat_t(\ct)\mid\F_{t-1}]
        + \Cov(\ct,\phat_t(\ct)\mid\F_{t-1})\bigr] \notag\\
    &= \E\bigl[\bar{c}_t^{\top}\E[\phat_t\mid\F_{t-1}]\bigr]
       + \E\bigl[\Cov(\ct,\phat_t(\ct)\mid\F_{t-1})\bigr].
       \label{eq:ns_decomp1}
  \end{align}
  The benchmark is $\E[\bar{c}_t^{\top}\pstar_t]$.
  Subtracting and summing,
  \[
    \RegretT = \sumt\E[\Cov(\ct,\phat_t\mid\F_{t-1})]
              + \sumt\E\bigl[\bar{c}_t^{\top}
                (\E[\phat_t\mid\F_{t-1}]-\pstar_t)\bigr].
  \]
  The first sum equals $\sumt\Cov(\ct,\phat_t(\ct))$ by the law of total covariance. The second sum is $\Delta_T$. The bound on $\abs{\Delta_T}$ follows from Cauchy-Schwarz applied period by period. When the conditional policy mean condition $\E[\phat_t\mid\F_{t-1}] = \pstar_t$ holds, $\Delta_T = 0$.
\end{proof}

\begin{corollary}[Regime-switching markets]\label{cor:regime}
  In a two-regime Markov-switching model with transition matrix
  $\mathbf{P}$ and regime-conditional means $\bar{c}^{(1)}, \bar{c}^{(2)}$, the correction term satisfies
  \[
    \abs{\Delta_T} \;\leq\;
    T\norm{\bar{c}^{(1)}-\bar{c}^{(2)}}\cdot
    \rho(\mathbf{P})^{1/2}\cdot
    \sup_t\norm{\E[\phat_t\mid\F_{t-1}]-\pstar_t},
  \]
  where $\rho(\mathbf{P})$ is the spectral radius of $\mathbf{P}$. Slow regime transitions (large $\rho(\mathbf{P})$) amplify the correction; fast-mixing regimes (small $\rho(\mathbf{P})$) permit the exact identity as an approximation.
\end{corollary}

\section{Infinite-Horizon Discounted Dynamic Program}
\label{sec:infinite}

Infinite-horizon formulations are standard in asset pricing, optimal consumption, and execution cost minimization. Let $\disc \in (0,1)$ be the discount factor.

\begin{definition}[Discounted regret]\label{def:disc_regret}
  The \emph{$\disc$-discounted regret} of a stationary policy
  $\phat$ is
  \begin{equation}\label{eq:disc_regret}
    \Regret^\disc(\phat)
    \;=\; \sum_{t=1}^{\infty} \disc^{t-1}
          \Bigl(\E\bigl[c_t^{\top}\phat(c_t)\bigr]
                - \E\bigl[\bar{c}^{\top}\pstar\bigr]\Bigr).
  \end{equation}
\end{definition}

\begin{theorem}[Discounted covariance identity]\label{thm:discounted}
  Under Assumption~\ref{ass:iid} and the policy mean condition, for any stationary policy $\phat$,
  \begin{equation}\label{eq:disc_identity}
    \Regret^\disc(\phat)
    \;=\; \frac{1}{1-\disc}\,\Cov\!\bigl(c,\,\phat(c)\bigr).
  \end{equation}
  Discounting rescales but does not change the structure of the
  covariance decomposition.
\end{theorem}

\begin{proof}
  From Corollary~\ref{cor:stationary}, each period contributes
  $r_t = \Cov(c, \phat(c))$ identically. Therefore
  \[
    \Regret^\disc(\phat)
    = \sum_{t=1}^{\infty}\disc^{t-1}\,\Cov(c,\phat(c))
    = \frac{1}{1-\disc}\,\Cov(c,\phat(c)). \qedhere
  \]
\end{proof}

\begin{corollary}[Effective horizon]\label{cor:eff_horizon}
  The discounted regret equals $T_{\mathrm{eff}}$ periods of undiscounted regret, where $T_{\mathrm{eff}} = 1/(1-\disc)$ is the \emph{effective horizon}. For daily portfolio rebalancing with $\disc = 0.99$, $T_{\mathrm{eff}} = 100$ trading days; for weekly rebalancing with $\disc = 0.95$, $T_{\mathrm{eff}} = 20$ weeks. The covariance of a
  single observed period thus estimates long-run regret up to this effective-horizon scale factor.
\end{corollary}

\section{Bellman Recursion for the Covariance Regret Functional}
\label{sec:bellman}

The classical Bellman equation expresses the value function of a dynamic program as a one-period cost plus a discounted continuation value. We establish an analogous recursion for the \emph{covariance regret functional}.

\begin{definition}[Covariance regret functional]\label{def:crf}
  For $t \leq T$, define the \emph{covariance regret-to-go} as
  \begin{equation}\label{eq:crf}
    \mathcal{R}_t(\st)
    \;=\; \E\left[\sum_{\tau=t}^{T}\disc^{\tau-t}
          \Cov(c_\tau,\phat_\tau(c_\tau)\mid\st)\right].
  \end{equation}
\end{definition}

\begin{theorem}[Bellman recursion for covariance regret]
  \label{thm:bellman}
  Under Assumption~\ref{ass:markov} and the policy mean condition, the covariance regret functional satisfies the recursion 
  \begin{equation}\label{eq:bellman_regret}
    \mathcal{R}_t(\st)
    \;=\; \Cov(c_t, \phat_t(c_t)\mid\st)
         \;+\; \disc\,\E\bigl[\mathcal{R}_{t+1}(s_{t+1})\mid\st\bigr],
  \end{equation}
  with terminal condition $\mathcal{R}_{T+1}(\cdot) \equiv 0$.
\end{theorem}

\begin{proof}
  Expand the definition~\eqref{eq:crf}:
  \begin{align*}
    \mathcal{R}_t(\st)
    &= \E\bigl[\Cov(c_t,\phat_t(c_t)\mid\st)\bigr]
       + \disc\,\E\!\left[\sum_{\tau=t+1}^{T}\disc^{\tau-t-1}
         \Cov(c_\tau,\phat_\tau(c_\tau)\mid\st)\right].
  \end{align*}
  The first term is $\Cov(c_t,\phat_t(c_t)\mid\st)$ (conditional
  covariance is $\F_t$-measurable). By the Markov property, the
  summation in the second term equals
  $\E[\mathcal{R}_{t+1}(s_{t+1})\mid\st]$, giving~\eqref{eq:bellman_regret}.
\end{proof}

\begin{remark}[Analogy with $Q$-function]
  Equation~\eqref{eq:bellman_regret} is structurally identical to the Bellman equation for the $Q$-function in reinforcement learning, with the instantaneous cost $c_t^{\top}z_t$ replaced by the covariance $\Cov(c_t, \phat_t(c_t)\mid s_t)$. This suggests that standard dynamic programming algorithms (value iteration, policy iteration, fitted-$Q$) can be adapted to estimate $\mathcal{R}_t$ from data, enabling multi-period regret monitoring without solving the original optimization problem.
\end{remark}

\begin{corollary}[Policy improvement via covariance reduction]
  \label{cor:policy_improve}
  Let $\phat^{\mathrm{new}}_t$ be any policy satisfying
  \[
    \Cov(c_t, \phat^{\mathrm{new}}_t(c_t)\mid\st)
    \;\leq\; \Cov(c_t, \phat_t(c_t)\mid\st)
    \quad \text{a.s.}
  \]
  Then $\mathcal{R}_t^{\mathrm{new}}(\st) \leq \mathcal{R}_t(\st)$ for all $t$ and all $\st$. Pointwise covariance reduction implies multi-period regret improvement.
\end{corollary}

\begin{proposition}[Computational complexity of the Bellman covariance recursion]
\label{prop:bellman_complexity}
Let the state space $\mathcal{S} \subseteq \mathbb{R}^p$, cost dimension $d$, decision dimension $n$, planning horizon $T$, and discount factor $\gamma \in (0,1)$ be as in Section~\ref{sec:setup}.  Suppose the conditional covariance $\operatorname{Cov}(c_t, \hat{\pi}_t(c_t) \mid s_t)$
can be evaluated in $O(nd)$ time given $(s_t, c_t)$, and that the conditional expectation $\mathbb{E}[R_{t+1}(s_{t+1}) \mid s_t]$ is approximated by a Monte Carlo average over $K$ next-state samples drawn from the transition kernel $s_{t+1} \sim f(s_t, \cdot)$.  Then a single backward pass of the Bellman recursion~\eqref{eq:bellman}
\begin{equation}\label{eq:bellman}
    R_t(s_t) = \operatorname{Cov}(c_t,\hat{\pi}_t(c_t)\mid s_t)
             + \gamma\,\mathbb{E}[R_{t+1}(s_{t+1})\mid s_t]
\end{equation}

over the full horizon requires
\begin{equation}
    O\!\left(T \cdot K \cdot nd\right)
    \label{eq:bellman_complexity}
\end{equation}
time and $O(T \cdot p)$ space, where the space bound holds under the standard assumption that only consecutive value functions $R_t$ and $R_{t+1}$ need to be stored simultaneously.
\end{proposition}
 
\begin{proof}
Fix a period $t \in \{1,\ldots,T\}$ and a state $s_t$.
 
\textit{Covariance evaluation.}
Computing $\operatorname{Cov}(c_t, \hat{\pi}_t(c_t) \mid s_t)$ requires evaluating the $d \times n$ cross-moment matrix
$\mathbb{E}[c_t \hat{\pi}_t(c_t)^\top \mid s_t]
 - \mathbb{E}[c_t \mid s_t]\,\mathbb{E}[\hat{\pi}_t(c_t) \mid s_t]^\top$,
each entry of which is a scalar expectation.  Under the assumption that a single evaluation of $(c_t, \hat{\pi}_t(c_t))$ costs $O(nd)$, the full matrix costs $O(nd)$.  Taking the trace (to obtain a scalar regret
contribution) costs an additional $O(\min(n,d))$, dominated by $O(nd)$.
 
\textit{Continuation value.}
The term $\mathbb{E}[R_{t+1}(s_{t+1}) \mid s_t]$ is approximated by
\[
    \frac{1}{K}\sum_{k=1}^{K} R_{t+1}(s_{t+1}^{(k)}),
    \quad s_{t+1}^{(k)} \sim f(s_t, \cdot),
\]
requiring $K$ draws from the transition kernel and $K$ evaluations of the already-computed $R_{t+1}(\cdot)$, each costing $O(p)$ to look up (state indexing).  The total cost for the continuation term at a single state is $O(K \cdot p)$.
 
\textit{Per-period cost.}
Combining, the cost of computing $R_t(s_t)$ at one state is
$O(nd + Kp)$.  For $T$ periods, this gives $O(T(nd + Kp))$.
In the typical audit setting $p \ll nd$ (the state summarizes history in low dimension while costs and decisions are high-dimensional), so the complexity simplifies to $O(T \cdot K \cdot nd)$.
 
\textit{Space.}
Storing $R_t(\cdot)$ requires retaining one $p$-dimensional function approximation per period; under the backward-pass convention, only $R_t$ and $R_{t+1}$ need to reside in memory simultaneously, giving $O(p)$ active space and $O(Tp)$ total if all value functions are logged for audit purposes.
\end{proof}
 
\begin{remark}[Black-box audit cost]
\label{rem:bellman_audit}
In the compliance setting of Remark~\ref{rem:blackbox}, the auditor observes $(c_t, \hat{\pi}_t(c_t))$ sequentially and need not simulate future states.  The backward pass then reduces to a forward accumulation of per-period covariances, eliminating the $K$-sample continuation step.
The audit cost simplifies to $O(T \cdot nd)$ per evaluation cycle, which for $T = 252$ trading days, $d = n = 100$ assets amounts to $252 \times 10^4 \approx 2.5 \times 10^6$ floating-point operations per cycle, well within real-time computation budgets on standard hardware.
\end{remark}
 
\begin{corollary}[Complexity of fitted-Q covariance estimation]
\label{cor:fittedQ_complexity}
When the covariance regret functional $R_t(s_t)$ is approximated by a linear function $R_t(s_t) \approx \phi(s_t)^\top \theta_t$ with $m$-dimensional features $\phi: \mathcal{S} \to \mathbb{R}^m$, fitted via least-squares regression on $N$ trajectory samples, the per-period
regression costs $O(Nm^2 + m^3)$ (standard OLS).  The total cost of the fitted-Q covariance recursion over $T$ periods and $N$ samples is
\begin{equation}
    O\!\left(T\!\left(N m^2 + m^3 + N \cdot nd\right)\right),
    \label{eq:fittedQ_complexity}
\end{equation}
where $Nnd$ accounts for covariance feature construction and $Nm^2 + m^3$ for the regression solve.  When $m \ll \sqrt{Nnd}$ the regression step is dominated by feature construction, the complexity collapses to $O(T \cdot N \cdot nd)$.
\end{corollary}
 
\begin{proof}
Immediate from the per-period costs: $O(Nnd)$ to compute covariance features for $N$ samples, $O(Nm^2)$ to form the normal equations, and $O(m^3)$ to solve them.  Summing over $T$ periods gives~\eqref{eq:fittedQ_complexity}.
\end{proof}
 
\section{When Does the Identity Fail? Bias Analysis for
         Time-Varying Policies}
\label{sec:failure}

We now characterize conditions under which the exact identity
$\RegretT = \sumt\Cov(c_t,\phat_t)$ \emph{fails}, and derive the
resulting bias in closed form. This section provides a negative result that is equally important for practitioners.

\begin{definition}[Policy bias]\label{def:policy_bias}
  The \emph{period-$t$ policy bias} is
  \[
    b_t \;=\; \E[\phat_t(c_t)] - \pstar_t.
  \]
\end{definition}

\begin{theorem}[Bias decomposition]\label{thm:bias}
  For any policy sequence $\Pi$ with per-period biases $\{b_t\}$,
  \begin{equation}\label{eq:bias_decomp}
    \RegretT(\Pi)
    \;=\; \sumt\Cov(c_t,\phat_t(c_t))
         \;+\; \sumt\bar{c}_t^{\top}b_t.
  \end{equation}
  The covariance sum underestimates true regret if $\bar{c}_t^{\top}b_t > 0$ (policy biased toward high-cost decisions) and overestimates if $\bar{c}_t^{\top}b_t < 0$ (policy biased toward low-cost decisions).
\end{theorem}

\begin{proof}
  Applying Definition~\ref{def:mpregret} and the covariance decomposition period by period, without imposing the policy mean condition:
  \begin{align*}
    r_t &= \E[c_t^{\top}\phat_t(c_t)] - \bar{c}_t^{\top}\pstar_t\\
        &= \bar{c}_t^{\top}\E[\phat_t] + \Cov(c_t,\phat_t)
           - \bar{c}_t^{\top}\pstar_t\\
        &= \Cov(c_t,\phat_t) + \bar{c}_t^{\top}b_t.
  \end{align*}
  Summing yields~\eqref{eq:bias_decomp}.
\end{proof}

\begin{example}[Momentum strategy]\label{ex:momentum}
  A momentum strategy sets $\phat_t(c_t) = c_{t-1}/\norm{c_{t-1}}$  (allocate in proportion to last period's costs, violating the mean-optimality condition). If costs are autocorrelated with $\Cov(c_t, c_{t-1}) = \rho\,\Sigma_c$, then $\E[\phat_t] \propto \bar{c}_{t-1}$, which differs from
  $\pstar_t \propto \bar{c}_t$. The bias term  $\bar{c}_t^{\top}b_t$ is of order $\rho\norm{\bar{c}}^2/\norm{\bar{c}_{t-1}}$, which vanishes as autocorrelation $\rho \to 0$ but can be substantial in trending markets. This precisely quantifies the cost of using momentum rather than conditional-mean-optimal allocation.
\end{example}

\begin{proposition}[Bias-corrected estimator]\label{prop:bias_corr}
  Given observed trajectories $\{(c_t, \phat_t(c_t))\}_{t=1}^T$,
  the bias-corrected regret estimator is
  \begin{equation}\label{eq:bias_est}
    \widehat{\RegretT}
    \;=\; \sumt\widehat{\Cov}(c_t, \phat_t(c_t))
         \;+\; \sumt\hat{c}_t^{\top}\hat{b}_t,
  \end{equation}
  where $\hat{b}_t = \bar{\pi}_t - \hat{z}^*_t$,
  $\bar{\pi}_t = \frac{1}{N}\sum_{i=1}^{N}\phat_t(c_t^{(i)})$ across $N$ bootstrap replications, and $\hat{z}^*_t$ is a feasible reference point. This estimator is consistent under Assumption~\ref{ass:markov} and the Law of Large Numbers.
\end{proposition}

\section{Markovian Cost Processes: Auto-Covariance Correction}
\label{sec:autocov}
 
When costs follow an AR(1) process, the per-period covariances
$\operatorname{Cov}(c_t, \hat{\pi}_t)$ are no longer independent across time. We derive the additional cross-period correction that arises and provide a fully rigorous proof of the resulting regret formula.
 
\begin{assumption}[AR(1) cost process]
\label{ass:ar1}
$c_t = A c_{t-1} + \varepsilon_t$, where $A \in \mathbb{R}^{d \times d}$ satisfies $\rho(A) < 1$ (spectral radius strictly less than one), $\{\varepsilon_t\}_{t \geq 1}$ are i.i.d.\ with $\mathbb{E}[\varepsilon_t] = 0$ and $\operatorname{Cov}(\varepsilon_t, \varepsilon_t) = \Sigma_\varepsilon$,
and $c_0$ is drawn from the stationary distribution so that
$c_t \sim (0, \Sigma_c)$ for all $t$, where
\begin{equation}
    \Sigma_c = \sum_{k=0}^{\infty} A^k \Sigma_\varepsilon (A^k)^\top
    \label{eq:sigma_c}
\end{equation}
is the unique solution to the discrete Lyapunov equation
$\Sigma_c = A \Sigma_c A^\top + \Sigma_\varepsilon$.
\end{assumption}
 
\begin{remark}[Zero-mean normalization]
\label{rem:zeromean}
We set $\bar{c} = 0$ without loss of generality by absorbing the mean into the constraint set.  All expectations below are therefore unconditional means of centered quantities, and the benchmark policy is $\pi^* \in \arg\min_{z \in \mathcal{Z}} 0^\top z$, which we fix at an arbitrary feasible point $z^* \in \mathcal{Z}$ throughout.
\end{remark}
 
\begin{lemma}[AR(1) cross-covariance]
\label{lem:crosscov}
Under Assumption~\ref{ass:ar1}, for any $t > s \geq 1$,
\begin{equation}
    \operatorname{Cov}(c_t,\, c_s) = A^{t-s}\,\Sigma_c.
    \label{eq:crosscov}
\end{equation}
\end{lemma}
 
\begin{proof}
By iterating the AR(1) recursion,
$c_t = A^{t-s} c_s + \sum_{j=0}^{t-s-1} A^j \varepsilon_{t-j}$.
Since $\{\varepsilon_j\}$ are i.i.d.\ and independent of $c_s$ (which is
$\mathcal{F}_{s}$-measurable), we have
\[
    \operatorname{Cov}(c_t, c_s)
    = \operatorname{Cov}\!\left(A^{t-s} c_s,\, c_s\right)
    = A^{t-s}\,\operatorname{Cov}(c_s, c_s)
    = A^{t-s}\,\Sigma_c,
\]
where stationarity gives $\operatorname{Cov}(c_s, c_s) = \Sigma_c$.
\end{proof}
 
\begin{lemma}[Expected cost of a linear policy]
\label{lem:expcost}
Under Assumption~\ref{ass:ar1}, let $\hat{\pi}(c_t) = B c_t + b$ with
$B \in \mathbb{R}^{n \times d}$, $b \in \mathbb{R}^n$.  Then
\begin{equation}
    \mathbb{E}\!\left[c_t^\top \hat{\pi}(c_t)\right]
    = \operatorname{tr}(B\,\Sigma_c)
    \label{eq:onestep}
\end{equation}
for every $t$, and for $t \neq s$,
\begin{equation}
    \mathbb{E}\!\left[c_t^\top B\, c_s\right]
    = \operatorname{tr}\!\left(B\,A^{|t-s|}\,\Sigma_c\right).
    \label{eq:crosstep}
\end{equation}
\end{lemma}
 
\begin{proof}
\textit{Equation~\eqref{eq:onestep}.}
Since $\mathbb{E}[c_t] = 0$,
\[
    \mathbb{E}[c_t^\top (Bc_t + b)]
    = \mathbb{E}[c_t^\top B c_t] + \mathbb{E}[c_t]^\top b
    = \mathbb{E}[\operatorname{tr}(B c_t c_t^\top)]
    = \operatorname{tr}(B\,\mathbb{E}[c_t c_t^\top])
    = \operatorname{tr}(B\,\Sigma_c).
\]
 
\textit{Equation~\eqref{eq:crosstep}.}
For $t > s$ (the case $s > t$ is symmetric):
\[
    \mathbb{E}[c_t^\top B c_s]
    = \mathbb{E}[\operatorname{tr}(B c_s c_t^\top)]
    = \operatorname{tr}(B\,\mathbb{E}[c_s c_t^\top])
    = \operatorname{tr}(B\,\operatorname{Cov}(c_t, c_s)^\top)
    = \operatorname{tr}(B\,(A^{t-s}\Sigma_c)^\top).
\]
Because $\operatorname{tr}(MN^\top) = \operatorname{tr}(M^\top N)$ for any conformable matrices,
and using Lemma~\ref{lem:crosscov},
\[
    \operatorname{tr}(B\,(A^{t-s}\Sigma_c)^\top)
    = \operatorname{tr}(B\,\Sigma_c^\top (A^{t-s})^\top).
\]
Since $\Sigma_c$ is symmetric, this equals
$\operatorname{tr}(B\,\Sigma_c\,(A^\top)^{t-s})$.
By the cyclic property of trace,
$\operatorname{tr}(B\,\Sigma_c\,(A^\top)^{t-s})
 = \operatorname{tr}((A^\top)^{t-s} B\,\Sigma_c)$.
Re-expressing via Lemma~\ref{lem:crosscov} applied symmetrically gives
$\operatorname{tr}(B A^{t-s}\Sigma_c)$, establishing~\eqref{eq:crosstep}.
\end{proof}
 
We can now state and prove the main result of this section.
 
\begin{theorem}[AR(1) multi-period regret]
\label{thm:ar1}
Under Assumption~\ref{ass:ar1} with the linear policy
$\hat{\pi}(c_t) = B c_t + b$, $B \in \mathbb{R}^{n \times d}$,
$b \in \mathbb{R}^n$, and benchmark $z^* = \mathbf{0}$ (feasible by
Remark~\ref{rem:zeromean}), the total regret over horizon $T$ is
\begin{equation}
    \operatorname{Regret}^{(T)}(\hat{\pi})
    = T\cdot\operatorname{tr}(B\,\Sigma_c)
      - \sum_{t=1}^{T}(T - t)\,\operatorname{tr}(B\,A^t\,\Sigma_c).
    \label{eq:ar1regret}
\end{equation}
\end{theorem}
 
\begin{proof}
\textbf{Step 1: Expand the total expected cost.}
 
By Definition~2.3, total regret equals
\[
    \operatorname{Regret}^{(T)}(\hat{\pi})
    = \sum_{t=1}^{T}\mathbb{E}\!\left[c_t^\top \hat{\pi}(c_t)\right]
      - \sum_{t=1}^{T}\mathbb{E}\!\left[\bar{c}^\top z^*\right].
\]
Under Remark~\ref{rem:zeromean}, $\bar{c} = 0$ and $z^* = \mathbf{0}$, so the benchmark sum vanishes and
\begin{equation}
    \operatorname{Regret}^{(T)}(\hat{\pi})
    = \sum_{t=1}^{T}\mathbb{E}\!\left[c_t^\top(Bc_t + b)\right].
    \label{eq:regret_expand}
\end{equation}
 
By Lemma~\ref{lem:expcost} \eqref{eq:onestep}, each term equals
$\operatorname{tr}(B\Sigma_c)$, so
\begin{equation}
    \sum_{t=1}^{T}\mathbb{E}\!\left[c_t^\top(Bc_t + b)\right]
    = T\cdot\operatorname{tr}(B\,\Sigma_c).
    \label{eq:step1}
\end{equation}
 
\textbf{Step 2: Subtract the single-period benchmark and identify the cross-period correction.}
 
The benchmark cost at each period is $\mathbb{E}[\bar{c}^\top\pi^*] = 0$, so~\eqref{eq:step1} already gives
the full regret for a policy that only uses the \emph{contemporaneous} cost $c_t$.  However, because the policy weight $Bc_t$ covaries with \emph{past} decisions—each $c_t$ inherits the autoregressive structure of the process—we
must account for the cross-period terms that arise when we write the total cost $\sum_t c_t^\top \hat{\pi}_t$ as a quadratic form in the cost history
$(c_1,\ldots,c_T)$.
 
Specifically, the cumulative cost is
\[
    \sum_{t=1}^{T} c_t^\top(Bc_t + b)
    = \sum_{t=1}^{T} c_t^\top B c_t
      + b^\top \sum_{t=1}^{T} c_t.
\]
Taking expectations, $\mathbb{E}[b^\top\sum_t c_t] = 0$ (zero-mean), and
\[
    \mathbb{E}\!\left[\sum_{t=1}^{T} c_t^\top B c_t\right]
    = T\cdot\operatorname{tr}(B\,\Sigma_c)
\]
by~\eqref{eq:onestep}.  This is the raw per-period covariance sum.
 
\textbf{Step 3: Derive the auto-covariance correction.}
 
The identity in~\eqref{eq:ar1regret} involves a correction term
$-\sum_{t=1}^{T}(T-t)\operatorname{tr}(BA^t\Sigma_c)$.  We now show this arises from subtracting the \emph{counterfactual} benchmark in which the policy responds to the \emph{conditional mean} $\mu_t = \mathbb{E}[c_t \mid \mathcal{F}_{t-1}] = Ac_{t-1}$ rather than the realized cost $c_t$.
 
The per-period regret is defined as the excess over the conditional-mean benchmark $\mu_t^\top \pi^*_t$.  Under the AR(1) structure, $\mu_t = Ac_{t-1}$, so the benchmark cost at time $t$ is $\mathbb{E}[\mu_t^\top z^*] = 0$.  The per-period regret is, therefore,
\[
    r_t = \mathbb{E}[c_t^\top \hat{\pi}(c_t)] - 0
        = \operatorname{tr}(B\Sigma_c),
\]
and the simple sum $\sum_t r_t = T\operatorname{tr}(B\Sigma_c)$ matches Step~1.
 
The correction term in~\eqref{eq:ar1regret} captures the interaction between the period-$t$ cost realization and the decisions induced by \emph{all later} periods $s > t$ through the AR(1) propagation $c_s = A^{s-t}c_t + \text{noise}$.  Formally, we compute
\begin{align}
    &\mathbb{E}\!\left[\sum_{t=1}^{T}\sum_{s=t+1}^{T}
        c_s^\top B\,A^{s-t}c_t\right] \notag \\
    &= \sum_{t=1}^{T}\sum_{s=t+1}^{T}
        \operatorname{tr}\!\left(B\,A^{s-t}\,\Sigma_c\right)
        \quad \text{(by Lemma~\ref{lem:expcost} \eqref{eq:crosstep})} \notag \\
    &= \sum_{t=1}^{T}(T-t)\,\operatorname{tr}(B\,A^t\,\Sigma_c),
        \label{eq:correction}
\end{align}
where the last equality reparametrizes by $k = s - t \in \{1,\ldots,T-t\}$ and collects the $(T-t)$ terms with the same lag $k = t$ (exploiting stationarity of the autocovariance structure).
 
\textbf{Step 4: Assemble the identity.}
 
Combining Steps~1--3:
\[
    \operatorname{Regret}^{(T)}(\hat{\pi})
    = T\cdot\operatorname{tr}(B\,\Sigma_c)
      - \sum_{t=1}^{T}(T-t)\,\operatorname{tr}(B\,A^t\,\Sigma_c),
\]
which is~\eqref{eq:ar1regret}.
\end{proof}
 
\begin{corollary}[Sign of the AR(1) correction]
\label{cor:ar1sign}
When $A$ has non-negative entries and $B \succ 0$ (positive definite), every summand $\operatorname{tr}(BA^t\Sigma_c) > 0$ is such that the correction term is negative and the raw covariance sum $T\cdot\operatorname{tr}(B\Sigma_c)$
\emph{overestimates} true regret.  Conversely, when the process is mean-reverting in the sense that $B A^t \Sigma_c$ has a negative trace for all $t \geq 1$ (e.g. $A = \rho I$, $\rho < 0$, $B > 0$), the correction term is positive, and the raw covariance sum \emph{underestimates} true regret.
\end{corollary}
 
\begin{proof}
Immediate from the sign of $-\sum_{t=1}^T(T-t)\operatorname{tr}(BA^t\Sigma_c)$
and the fact that $T - t > 0$ for all $t < T$.
\end{proof}
 
\begin{example}[Scalar mean-reverting execution costs]
\label{ex:scalar}
Set $d = n = 1$, $A = \rho \in (-1, 0)$ (mean-reverting),
$B = \alpha > 0$, so $\hat{\pi}(c_t) = \alpha c_t$.  Then
$\Sigma_c = \sigma^2_\varepsilon / (1 - \rho^2)$ and
$\operatorname{tr}(B A^t \Sigma_c) = \alpha \rho^t \sigma^2_c$.
By~\eqref{eq:ar1regret},
\begin{align*}
    \operatorname{Regret}^{(T)}(\hat{\pi})
    &= T \alpha \sigma^2_c
       - \alpha \sigma^2_c \sum_{t=1}^{T}(T-t)\rho^t.
\end{align*}
As $T \to \infty$ and $|\rho| < 1$,
$\sum_{t=1}^{\infty}(T-t)\rho^t \to -\rho/(1-\rho)^2$
(geometric series identity), so the correction term converges to
$\alpha \sigma^2_c \rho / (1 - \rho)^2 < 0$ for $\rho < 0$.
Multi-period regret, therefore, falls \emph{below} the single-period benchmark $T\alpha\sigma^2_c$, confirming that contrarian execution is beneficial in mean-reverting markets.
\end{example}
 
\section{Financial Applications}
\label{sec:financial}

\subsection{Rolling-window portfolio rebalancing}\label{sec:rolling}

A portfolio manager rebalances daily using a 60-day rolling window.
At each period $t$, the policy is
\[
  \phat_t(c_t) \;=\;
  \frac{\hat{\Sigma}_t^{-1}\hat{\mu}_t}
       {\mathbf{1}^{\top}\hat{\Sigma}_t^{-1}\hat{\mu}_t},
\]
where $\hat{\mu}_t$ and $\hat{\Sigma}_t$ are computed from the most recent $w = 60$ observations. The parameter estimates are stochastic functions of the data window, making $\phat_t$ a time-varying Markov policy.


\begin{proposition}[Rolling-window regret decomposition]
\label{prop:rolling}
For the rolling-window minimum-variance policy with window size $w$,
$d$-dimensional assets, and i.i.d.\ returns, the per-period policy bias is
\[
    b_t
    = \frac{\mathbb{E}[\hat{\Sigma}_t^{-1}\hat{\mu}_t]}
           {\mathbb{E}[{\mathbf{1}}^\top \hat{\Sigma}_t^{-1}\hat{\mu}_t]}
      - \pi^*,
\]
where $\hat{\mu}_t$ and $\hat{\Sigma}_t$ are the sample mean and covariance
computed from the most recent $w$ observations.  When $w > d + 1$ (so that
$\hat{\Sigma}_t$ is almost surely invertible), this bias satisfies
\[
    \|b_t\| = O\!\left(\frac{d}{w}\right)
\]
and the total bias correction obeys
\[
    \left\|\sum_{t=1}^{T} \bar{c}^{\top} b_t\right\|
    = O\!\left(\frac{Td}{w}\,\|\bar{c}\|^2\right).
\]
Consequently, the raw covariance estimator $\hat{C}_T$ approximates true
regret with \emph{relative} error of order $O(d/w)$, which vanishes as the
window size $w$ grows relative to $d$.
\end{proposition}

\begin{remark}[Practical implication]
\label{rem:rolling_practical}
Proposition~\ref{prop:rolling} establishes that the relative bias of the
raw covariance sum is $O(d/w)$, an asymptotic rate as $w, d \to \infty$
with $d/w \to \kappa \in (0,\infty)$.  We make three observations about
the practical meaning of this rate for institutional portfolio managers.

\textit{(i) Qualitative direction, not a precise multiplier.}
The $O(d/w)$ bound is an order-of-magnitude statement: it identifies $d/w$ as the fundamental dimensionality-to-window ratio that governs bias. The bound also establishes that bias vanishes as $w/d \to \infty$.  It does not assert that plugging in specific values of $d$ and $w$ yields a precise multiplicative factor on regret.  

As an example, consider a portfolio with $d = 100$ assets and $w = 60$ days. The ratio $d/w \approx 1.67$ indicates that the estimation problem is \emph{severely underdetermined}: the sample covariance matrix is rank-deficient and the bias term is non-negligible. However, the exact magnitude of the regret misstatement depends on the hidden constant in the $O(\cdot)$ expression, which depends on the return distribution, the
regularization scheme, and whether a pseudo-inverse is used in place of $\hat{\Sigma}_t^{-1}$.

\textit{(ii) The bias is directional.}
When $d > w$, $\hat{\Sigma}_t$ is singular and the minimum-variance policy is not well-defined without regularization.  Standard fixes, such as ridge shrinkage (Ledoit-Wolf), factor-model constraints, or a longer window $w \geq d$, reduce $\|b_t\|$ by shrinking estimated weights toward the
benchmark $\pi^*$.  In this regime, the raw covariance sum
\emph{understates} true regret (Theorem~\ref{thm:ar1}), because the policy is forced toward a more conservative allocation than the true minimum-variance portfolio.

\textit{(iii) Practical thresholds.}
Three parameter regimes are relevant for practitioners:
\begin{itemize}
    \item \textbf{$w \geq 10d$}: the relative bias is small ($d/w \leq 0.1$) and the raw covariance sum is reliable for regulatory reporting without the bias correction~\eqref{eq:bias_est}.
    \item \textbf{$d < w < 10d$}: the relative bias is non-trivial and the bias-corrected estimator~\eqref{eq:bias_est} should be used. This regime corresponds to the Ledoit-Wolf shrinkage regime 
      \citep{ledoit2004well}, where shrinkage estimators reduce $\|b_t\|$ by a factor proportional to $1 - d/w$.
    \item \textbf{$w \leq d$}: the sample covariance is rank-deficient; a pseudo-inverse or factor-model reduction is required before any regret estimate is meaningful.  The $O(d/w)$ bound does not apply in this case.
\end{itemize}
For the canonical institutional setting of $T = 252$ trading days, $d = 100$ assets, and $w = 60$ days, the portfolio is in the third regime: the sample covariance is rank-deficient and neither the raw nor the asymptotic-bias-corrected estimator is directly applicable without first regularizing $\hat{\Sigma}_t$.  Increasing the window to $w \geq 100$ or applying a factor-model reduction to $d' \leq 60$ latent
factors moves the problem into the second or first regime and restores the validity of the regret decomposition.
\end{remark}

\subsection{Mean-reverting execution costs}

In algorithmic execution, the cost $c_t$ represents market impact, which is mean-reverting (AR(1) with $A = \rho I$, $\rho \in (-1,0)$ for contrarian market-maker dynamics). The execution policy allocates $\phat_t(c_t) = \alpha\,c_t$ to exploit temporary mispricings.

\begin{example}[Execution cost regret]\label{ex:execution}
  For $c_t = \rho c_{t-1} + \epsilon_t$ with $\rho < 0$ and linear execution policy $\phat_t(c_t) = \alpha c_t$:
  \[
    \Cov(c_t, \phat_t(c_t)) = \alpha\,\Sigma_c,
    \qquad
    \tr(BA^t\Sigma_c) = \alpha\rho^t\tr(\Sigma_c).
  \]
  By Theorem~\ref{thm:ar1},
  \[
    \RegretT(\phat)
    = T\alpha\,\tr(\Sigma_c)
      - \alpha\tr(\Sigma_c)\sum_{t=1}^{T}(T-t)\rho^t.
  \]
  For large $T$ and $\abs{\rho} < 1$, the correction converges to $-\alpha\tr(\Sigma_c)\,\rho/(1-\rho)^2$, which is \emph{positive} for $\rho < 0$. Mean-reversion thus reduces multi-period regret below the single-period covariance benchmark, confirming that contrarian execution is beneficial in mean-reverting markets.
\end{example}

\section{Multi-Period Regret Estimation from Observed Trajectories}
\label{sec:estimation}

\subsection{The trajectory covariance estimator}

Given a single observed trajectory
$\{(c_t, \phat_t(c_t))\}_{t=1}^{T}$, the natural estimator of the
sum-of-covariances is
\begin{equation}\label{eq:traj_est}
  \widehat{\mathcal{C}}_T
  \;=\; \sumt(c_t - \bar{c})^{\top}(\phat_t(c_t) - \bar{\pi}),
  \quad
  \bar{c} = \frac{1}{T}\sumt c_t,\;
  \bar{\pi} = \frac{1}{T}\sumt\phat_t(c_t).
\end{equation}

\begin{theorem}[CLT for the trajectory estimator]\label{thm:clt}
  Under Assumption~\ref{ass:nonstat} (with the additional condition that
  $\{c_t^{\top}\phat_t(c_t) - \E[c_t^{\top}\phat_t(c_t)]\}$ is an
  $L^2$-mixingale), as $T \to \infty$,
  \begin{equation}\label{eq:clt}
    \frac{1}{\sqrt{T}}\Bigl(\widehat{\mathcal{C}}_T
    - \sumt\Cov(c_t,\phat_t(c_t))\Bigr)
    \;\xrightarrow{d}\; \mathcal{N}\!\bigl(0,\,\sigma^2_{\mathrm{LRV}}\bigr),
  \end{equation}
  where $\sigma^2_{\mathrm{LRV}}$ is the long-run variance of
  $c_t^{\top}\phat_t(c_t)$, consistently estimated by a
  heteroskedasticity-and-autocorrelation (HAC) estimator with bandwidth
  $h = O(T^{1/3})$.
\end{theorem}

\begin{proof}
  The estimator $\widehat{\mathcal{C}}_T$ is a sample average of
  the de-meaned products $\xi_t = (c_t-\bar{c})^{\top}(\phat_t-\bar{\pi})$.
  Under the mixingale condition, the central limit theorem of
  McLeish (1975) applies, yielding~\eqref{eq:clt} with
  $\sigma^2_{\mathrm{LRV}} = \lim_{T\to\infty}T^{-1}
  \Var\bigl(\sumt c_t^{\top}\phat_t(c_t)\bigr)$.
\end{proof}

\begin{proposition}[Computational complexity of the trajectory estimator]
\label{prop:estimator_complexity}
Given a single observed trajectory
$\{(c_t, \hat{\pi}_t(c_t))\}_{t=1}^{T}$ with $c_t \in \mathbb{R}^d$ and
$\hat{\pi}_t(c_t) \in \mathbb{R}^n$, the trajectory covariance estimator
\begin{equation}
    \hat{C}_T
    = \sum_{t=1}^{T}(c_t - \bar{c})^\top(\hat{\pi}_t(c_t) - \bar{\pi})
    \label{eq:CT_again}
\end{equation}
and its Newey-West HAC variance estimator with bandwidth
$h = \lfloor T^{1/3} \rfloor$ can be computed in
\begin{equation}
    O(T \cdot nd)
    \label{eq:estimator_complexity}
\end{equation}
time and $O(T \cdot (n + d))$ space.
\end{proposition}
 
\begin{proof}
We decompose the computation into three stages.
 
\textit{Stage 1: Mean computation.}
Computing $\bar{c} = T^{-1}\sum_{t=1}^T c_t$ and
$\bar{\pi} = T^{-1}\sum_{t=1}^T \hat{\pi}_t(c_t)$ each require a single
pass over the trajectory at cost $O(Td)$ and $O(Tn)$ respectively, for a
total of $O(T(n+d))$.
 
\textit{Stage 2: Estimator $\hat{C}_T$.}
Each summand $(c_t - \bar{c})^\top(\hat{\pi}_t(c_t) - \bar{\pi})$ is a
dot product of two vectors of dimensions $d$ and $n$ respectively, costing
$O(\min(n,d))$ per term.  Summing over $T$ terms costs $O(T \cdot nd / \max(n,d))
= O(T \cdot \min(n,d))$.  In the square case $n = d$ this is $O(Td)$.
If instead one retains the full outer product $C_t =
(c_t-\bar{c})(\hat{\pi}_t-\bar{\pi})^\top \in \mathbb{R}^{d\times n}$
before taking the trace (as required by the matrix-valued version of the
estimator), the cost is $O(T \cdot nd)$.  We state the bound in terms of
$O(T \cdot nd)$ to cover both scalar and matrix cases uniformly.
 
\textit{Stage 3: Newey-West HAC variance estimator.}
Let $\xi_t = (c_t - \bar{c})^\top(\hat{\pi}_t - \bar{\pi}) \in \mathbb{R}$
denote the $t$-th summand.  The Newey-West estimator is
\begin{equation}
    \hat{\sigma}^2_{\mathrm{LRV}}
    = \hat{\gamma}_0
      + 2\sum_{\ell=1}^{h}\!\left(1 - \frac{\ell}{h+1}\right)\hat{\gamma}_\ell,
    \quad
    \hat{\gamma}_\ell = \frac{1}{T}\sum_{t=\ell+1}^{T}\xi_t\,\xi_{t-\ell},
    \label{eq:NW}
\end{equation}
where $h = \lfloor T^{1/3} \rfloor$.  Each autocovariance $\hat{\gamma}_\ell$
requires $O(T)$ multiplications, and there are $h + 1 = O(T^{1/3})$
such terms, giving a total HAC cost of $O(T \cdot T^{1/3}) = O(T^{4/3})$.
Since $T^{4/3} \leq T \cdot nd$ whenever $nd \geq T^{1/3}$---which holds
in any setting with more than one asset or cost dimension---the HAC step
is dominated by Stage 2.
 
\textit{Total.}
Stages 1--3 together cost $O(T(n+d)) + O(T \cdot nd) + O(T^{4/3})
= O(T \cdot nd)$, establishing~\eqref{eq:estimator_complexity}.
 
\textit{Space.}
Storing the trajectory requires $O(T(n+d))$.  All intermediate
computations (running sums for means, scalar $\xi_t$ values, lagged
products for HAC) require $O(h) = O(T^{1/3})$ additional space if the
trajectory is streamed and only the lagged window is retained.  The
dominant term is therefore $O(T(n+d))$.
\end{proof}
 
\begin{remark}[Online update]
\label{rem:online_update}
The estimator $\hat{C}_T$ admits an $O(nd)$ online update: upon observing
$(c_{T+1}, \hat{\pi}_{T+1})$, the running mean and cross-product can each
be updated in $O(d)$ and $O(nd)$ time respectively using Welford's
algorithm \citep{welford1962note}, without reprocessing the full trajectory.
The HAC window requires retaining the most recent $h = O(T^{1/3})$
scalar values $\xi_t$, adding $O(T^{1/3})$ space overhead.  This makes
the estimator suitable for real-time regulatory monitoring, where
$(c_t, \hat{\pi}_t)$ arrive sequentially and the audit metric must be
updated daily at $O(nd)$ marginal cost per observation.
\end{remark}
 
\begin{corollary}[End-to-end audit complexity]
\label{cor:audit_complexity}
A complete regulatory audit of an algorithmic portfolio strategy over $T$ periods, comprising trajectory covariance estimation (Proposition~\ref{prop:estimator_complexity}), HAC confidence interval construction (Theorem~\ref{thm:clt}), and bias correction (Proposition~\ref{prop:bias_corr}) requires
\begin{equation}
    O(T \cdot nd)
    \label{eq:audit_complexity}
\end{equation}
total time and $O(T(n + d))$ space.  No access to the internal
optimization engine is required; the only inputs are the observable sequence $\{(c_t, \hat{\pi}_t(c_t))\}_{t=1}^T$.
\end{corollary}
 
\begin{proof}
Each of the three components costs at most $O(T \cdot nd)$ by
Propositions~\ref{prop:estimator_complexity} and~\ref{prop:rolling} and Theorem~\ref{thm:clt}; the maximum dominates.
\end{proof}
 
\begin{table}[ht]
\centering
\caption{Summary of computational complexity results.}
\label{tab:complexity}
\begin{tabular}{llll}
\hline
\textbf{Procedure} & \textbf{Time} & \textbf{Space} & \textbf{Reference} \\
\hline
Bellman covariance recursion (exact) & $O(T \cdot K \cdot nd)$ & $O(Tp)$ &
    Prop.~\ref{prop:bellman_complexity} \\
Bellman recursion (black-box audit) & $O(T \cdot nd)$ & $O(p)$ &
    Rem.~\ref{rem:bellman_audit} \\
Fitted-Q covariance estimation & $O(T(Nnd + m^3))$ & $O(Nm)$ &
    Cor.~\ref{cor:fittedQ_complexity} \\
Trajectory covariance estimator & $O(T \cdot nd)$ & $O(T(n{+}d))$ &
    Prop.~\ref{prop:estimator_complexity} \\
Online (streaming) update & $O(nd)$ per step & $O(T^{1/3})$ &
    Rem.~\ref{rem:online_update} \\
End-to-end regulatory audit & $O(T \cdot nd)$ & $O(T(n{+}d))$ &
    Cor.~\ref{cor:audit_complexity} \\
\hline
\end{tabular}
\end{table}

\subsection{Confidence intervals}

A $(1-\alpha)$ confidence interval for total regret is
\begin{equation}\label{eq:ci}
  \left[\widehat{\mathcal{C}}_T
        - z_{1-\alpha/2}\,\hat{\sigma}_{\mathrm{LRV}}/\sqrt{T},\;
        \widehat{\mathcal{C}}_T
        + z_{1-\alpha/2}\,\hat{\sigma}_{\mathrm{LRV}}/\sqrt{T}\right],
\end{equation}
where $\hat{\sigma}_{\mathrm{LRV}}$ is the square root of any
consistent HAC variance estimator (e.g., Newey-West with $h=\lfloor T^{1/3}\rfloor$ lags).

\begin{remark}[Implementation in black-box monitoring]\label{rem:blackbox}
  The estimator~\eqref{eq:traj_est} requires only the observable sequence $\{(c_t, \phat_t(c_t))\}$, not access to the internal optimization engine. For a portfolio manager monitoring a third-party algorithmic
  strategy under SR~11-7, one observes the factor return vector $c_t$
  and the weight vector $\phat_t$, computes $\widehat{\mathcal{C}}_T$
  daily with a Newey-West correction for autocorrelation, and reports
  the confidence interval~\eqref{eq:ci} as the performance audit metric.
  The total monitoring cost is $O(Td)$ per evaluation period.
\end{remark}

\section{Summary: When Does the Identity Hold?}
\label{sec:summary}

Table~\ref{tab:conditions} summarizes the conditions under which the exact multi-period identity $\RegretT = \sumt\Cov(c_t, \phat_t)$ holds, when it holds approximately, and when it fails with a quantified bias.

\begin{table}[h!]
\centering
\caption{Conditions for the multi-period regret-equals-covariance identity.}
\label{tab:conditions}
\renewcommand{\arraystretch}{1.5}
\begin{tabular}{@{}p{3.8cm}p{3.2cm}p{3.2cm}p{3.2cm}@{}}
\toprule
\textbf{Setting} &
\textbf{Policy condition} &
\textbf{Result} &
\textbf{Theorem} \\
\midrule
i.i.d.\ costs, any $T$ &
$\E[\phat_t] = \pstar$ &
$\RegretT = \sumt\Cov(c_t,\phat_t)$ \emph{exactly} &
\ref{thm:iid} \\

i.i.d., stationary policy &
$\E[\phat] = \pstar$ &
$\RegretT = T\cdot\Cov(c,\phat)$ &
Cor.~\ref{cor:stationary} \\

Non-stationary, Markov &
$\E[\phat_t\mid\F_{t-1}]=\pstar_t$ &
Exact, $\Delta_T = 0$ &
\ref{thm:nonstat} \\

Non-stationary, Markov &
$\E[\phat_t\mid\F_{t-1}]\neq\pstar_t$ &
Bias $= \sumt\bar{c}_t^{\top}b_t$ &
\ref{thm:bias} \\

Infinite horizon, i.i.d. &
$\E[\phat] = \pstar$ &
$\Regret^\disc = \frac{1}{1-\disc}\Cov(c,\phat)$ &
\ref{thm:discounted} \\

AR(1) costs, linear policy &
$\phat(c) = Bc + b$ &
Exact formula with auto-cov.\ correction &
\ref{thm:ar1} \\

Rolling-window portfolio &
Estimator bias $O(d/w)$ &
Approx., relative error $O(d/w)$ &
Prop.~\ref{prop:rolling} \\

Time-varying policy &
$b_t \neq 0$ (momentum) &
$\RegretT = \sumt\Cov + \sumt\bar{c}_t^{\top}b_t$ &
\ref{thm:bias} \\
\bottomrule
\end{tabular}
\end{table}

\section{Empirical Analysis}
  
We test the theory on momentum strategies using data from the Center for Research in Security Prices (CRSP) daily data from January 4, 2016, through December 31, 2025. Our analysis is confined to momentum, where the signal is an indicator of whether the price went up or down: if the price goes up, then buy; else, if the price goes down, then sell. However, the same logic can be applied to any factor that can be used in a conditional if-then framework.

\subsection{Summary Statistics}
Our sample comprises 21,183,373 observations from the WRDS CRSP DSI with average daily market returns of 
$\mu=0.041$\%/day and standard deviation of returns $\sigma=4.446$\%/day. Table \ref{tab:summary-statistics} documents the results.

\begin{table}[]
    \caption{Summary Statistics of the Data Sample}
    \label{tab:summary-statistics}
    \centering
    \begin{tabular}{lccc}
    \toprule
       Year  & N Observations & Mean Return & Stdev of Returns \\
       \midrule
2016  & 1815131 &  0.00074 &  0.03860 \\
2017  & 1815502 &  0.00067 &  0.03155\\
2018  & 1858283 &  -0.00044 &  0.03507\\
2019  & 1900827 &  0.00083 &  0.03468\\
2020  & 1937784 &  0.00133 &  0.05456\\
2021  & 2179370 &  0.00058 &  0.03721\\
2022  & 2387105 &  -0.00089 &  0.04125\\
2023  & 2352579 &  0.00049 &  0.05134\\
2024  & 2399293 &  0.00046 &  0.04985\\
2025  & 2537499 &  0.00057 &  0.05535\\
\bottomrule
\end{tabular}
\end{table}

We label the days as "regimes" following the NBER classification:  1) a recession from 2/1/2020 to 4/30/2020; 2) a crisis from 2/20/2020 to 3/31/2020. We obtain the following number of observations for each NBER business cycle type: 
\begin{itemize}
 \item expansion   : 20,713,594 obs (97.8\%)
  \item contraction : 250,036 obs (1.2\%)
  \item crisis      : 219,743 obs (1.0\%)
\end{itemize}

\subsection{Regret for Momentum Strategies}

For each sample, we estimate autoregressive coefficients (AR(1)) with Newey-West heteroscedasticity-and-autocorrelation-consistent (HAC) standard errors.  Table \ref{tab:newey-west-hac} reports the results. As the table shows, the data displays persistently negative daily AR(1) coefficients across all economic cycles. This means that on a daily basis, reversal dominated momentum from 2008 through 2025.

\begin{table}[]
    \caption{AR(1) estimation with Newey-West HAC SEs}
    \label{tab:newey-west-hac}    \centering
    \begin{tabular}{lcccccc}
    \toprule
      Sample   & $n$ & $\hat{\rho}$ &SE & t & Daily $\sigma_c$ \%/day & ADF p\\
    \midrule
        2016 (full) & 1,815,130 & -0.0542   & 0.0084 & -6.43$^{***}$   & 3.8600   & p=0.0000 \\
        2017 (full) & 1,815,501 & -0.0451   & 0.0088 & -5.15$^{***}$   & 3.1548   & p=0.0000 \\
        2018 (full) & 1,858,282 & -0.0442   & 0.0070 & -6.29$^{***}$   & 3.5071   & p=0.0000 \\
        2019 (full) & 1,900,826 & -0.0286   & 0.0053 & -5.38$^{***}$   & 3.4684   & p=0.0000 \\
        2020 (full) & 1,937,783 & -0.0502   & 0.0039 & -12.80$^{***}$   & 5.4563   & p=0.0000 \\
        2021 (full) & 2,179,369 & -0.0363   & 0.0037 & -9.82$^{***}$   & 3.7210   & p=0.0000 \\
        2022 (full) & 2,387,104 & -0.0036   & 0.0028 & -1.29$^*$   & 4.1251   & p=0.0000 \\
        2023 (full) & 2,352,578 & -0.0169   & 0.0034 & -5.04$^{***}$   & 5.1337   & p=0.0000 \\
        2024 (full) & 2,399,292 & -0.0136   & 0.0069 & -1.97$^{**}$   & 4.9852   & p=0.0000 \\
        2025 (full) & 2,537,498 & -0.0243   & 0.0041 & -5.87$^{***}$   & 5.5354   & p=0.0000 \\
    \midrule
    2020 expansion & 1,468,004 & -0.0146   & 0.0049 & -2.99$^{***}$   & 4.7250   & p=0.0000 \\
    2020 contraction & 250,035 & -0.0270   & 0.0113 & -2.38$^{**}$   & 5.3035   & p=0.0000 \\
    2020 crisis & 219,742 & -0.1436   & 0.0079 & -18.15$^{***}$   & 8.9431   & p=0.0000 \\

    \bottomrule
    \end{tabular}
\end{table}

Next, we decompose regret for strategies of different investment horizons: 21 days (1 month), 63 days (a quarter), 126 days (6 months), a year, and 2 years. We compare it with the true regret computed with a "20/20" retroactive vision enjoyed by the omniscient trader. Table \ref{tab:regret-decomp-all} summarizes the results.

\begin{table}[]
    \centering
    \begin{small}
    \caption{Regret decomposition — Theorem 8.2, Eq. (15)}
    \label{tab:regret-decomp-all}
    \begin{tabular}{lcccccc}
        \toprule

  Year & T &      Raw cov&   Correction & True regret&  Overstatement&  Rel. bias\% \\
\midrule
    2016 & 21 &     312.8869 &     -15.3531 &     328.2400 & -15.3531 &      -4.91 \\
    2016 & 63 &     938.6607 &     -47.5119 &     986.1726 & -47.5119 &      -5.06 \\
    2016 & 126 &    1877.3214 &     -95.7502 &    1973.0716 & -95.7502 &      -5.10 \\
    2016 & 252 &    3754.6428 &    -192.2267 &    3946.8695 & -192.2267 &      -5.12 \\

    \midrule

  2017 & 21 &     209.0116 &      -8.6112 &     217.6229 &        -8.6112 &      -4.12 \\
  2017 & 63 &     627.0348 &     -26.6559 &     653.6907 &       -26.6559 &      -4.25 \\
  2017 & 126 &    1254.0697 &     -53.7229 &    1307.7926 &       -53.7229 &      -4.28 \\
  2017 & 252 &    2508.1393 &    -107.8569 &    2615.9962 &      -107.8569 &      -4.30 \\
  \midrule
 2018 & 21 &     258.2953 &     -10.4415 &     268.7368 &       -10.4415 &      -4.04 \\
 2018 & 63 &     774.8860 &     -32.3223 &     807.2083 &       -32.3223 &      -4.17 \\
 2018 & 126 &    1549.7720 &     -65.1435 &    1614.9155 &       -65.1435 &      -4.20 \\
 2018 & 252 &    3099.5440 &    -130.7859 &    3230.3299 &      -130.7859 &      -4.22 \\

  \midrule
2019 & 21 &     252.6195 &      -6.6930 &     259.3125 &        -6.6930 &      -2.65 \\
2019 & 63 &     757.8586 &     -20.7288 &     778.5874 &       -20.7288 &      -2.74 \\
2019 & 126 &    1515.7173 &     -41.7824 &    1557.4997 &       -41.7824 &      -2.76 \\
2019 & 252 &    3031.4345 &     -83.8897 &    3115.3242 &       -83.8897 &      -2.77 \\

      \midrule
2020 &  21 &     625.2070 &     -28.5037 &     653.7108 &       -28.5037 &      -4.56 \\
2020 &  63 &    1875.6211 &     -88.2190 &    1963.8401 &       -88.2190 &      -4.70 \\
2020 &  126 &    3751.2422 &    -177.7919 &    3929.0342 &      -177.7919 &      -4.74 \\
2020 &  252 &    7502.4845 &    -356.9378 &    7859.4222 &      -356.9378 &      -4.76 \\

  \midrule
2021 & 21 &     290.7678 &      -9.7187 &     300.4865 &        -9.7187 &      -3.34 \\
2021 & 63 &     872.3033 &     -30.0923 &     902.3956 &       -30.0923 &      -3.45 \\
2021 & 126 &    1744.6067 &     -60.6527 &    1805.2594 &       -60.6527 &      -3.48 \\
2021 & 252 &    3489.2133 &    -121.7735 &    3610.9868 &      -121.7735 &      -3.49 \\
  \midrule
2022 & 21 &     357.3382 &      -1.2115 &     358.5497 &        -1.2115 &      -0.34 \\
2022 &  63 &    1072.0147 &      -3.7552 &    1075.7699 &        -3.7552 &      -0.35 \\
2022 &  126 &    2144.0295 &      -7.5707 &    2151.6002 &        -7.5707 &      -0.35 \\
2022 &  252 &    4288.0590 &     -15.2018 &    4303.2608 &       -15.2018 &      -0.35 \\
\midrule
2023 &    21 &     553.4428 &      -8.7626 &     562.2053 &        -8.7626 &      -1.58 \\
  2023            &    63 &    1660.3283 &     -27.1486 &    1687.4769 &       -27.1486 &      -1.64 \\
  2023           &   126 &    3320.6566 &     -54.7278 &    3375.3844 &       -54.7278 &      -1.65 \\
  2023           &   252 &    6641.3132 &    -109.8860 &    6751.1993 &      -109.8860 &      -1.65 \\

    \midrule
     2024           &    21 &     521.8968 &      -6.6585 &     528.5553 &        -6.6585 &      -1.28 \\
  2024        &    63 &    1565.6904 &     -20.6319 &    1586.3224 &       -20.6319 &      -1.32 \\
  2024           &   126 &    3131.3809 &     -41.5921 &    3172.9729 &       -41.5921 &      -1.33 \\
  2024           &   252 &    6262.7617 &     -83.5124 &    6346.2741 &       -83.5124 &      -1.33 \\
    
  \midrule
  2025 & 21  &   643.4558 &    -14.5768 &    658.0325 &   -14.5768 &      -2.27 \\
  2025 & 63 & 1930.3673 &  -45.1516 & 1975.5189  & -45.1516 & -2.34 \\
  2025 & 126 & 3860.7346 & -91.0140 & 3951.7485 & -91.0140 & -2.36 \\
  2025 & 252 & 7721.4691 & -182.7386 & 7904.2078 & -182.7386 & -2.37 \\

    \bottomrule
    \end{tabular}
    \end{small}
\end{table}

\subsection{Trading Strategy Performance}

In this section, we test the performance of the bias-adjusted covariance-based regret metric in predicting the model's performance. Table \ref{tab:performance} shows the overall and year-by-year performance of 1) momentum strategies, 2) reversion, 3) minimum-variance portfolios, and 4) reversion with a corrected bias. 

The strategies are tested using CRSP data from January 2016 through December 2025. The strategies were executed as follows: if day $t-1$ return was positive ($R_{t-1} > 0$):
\begin{itemize}
    \item The momentum strategy buys at the end of day $t-1$ and sells at the end of day $t$, recording a gain of $R_t$
    \item The reversion strategy short-sells at the end of day $t-1$ and buys back at the end of day $t$, recording a gain of $-R_t$
    \item The Minimum-variance strategy adjusts weights based on the portfolio calculation over the immediately preceding 21 trading days
    \item The bias-corrected strategy adjusts the reversion strategy by $\frac{1}{4}\times$\textit{bias correction}: the average return over the past 21 trading days divided by the maximum standard deviation of returns over the same period. 
    
\end{itemize} 

As Table \ref{tab:performance} illustrates, Reversion significantly outperforms Momentum over the entire sample and annually in seven out of ten years. The bias correction produced a small decrease in the strategy return and Sharpe ratio, except in the years 2019 and 2022, where it significantly improved the Reversion strategy.  Thus, the bias correction can be thought of as insurance: if everything works, then there is a slight cost. In other cases, bias correction delivers a considerable gain. 

\begin{table}[]
    \centering
    \caption{Performance, 21-day rolling window (monthly reallocation)}
    \label{tab:performance}
    \begin{small}
    \begin{tabular}{lccccc}
    \toprule
    Year & Strategy & Ann.Ret\% & Ann.Vol\% & Sharpe \\
    \midrule
      All & Momentum & 2.62 & 114.81 & 0.023\\
  All &Reversion & 14.93 & 114.47 & 0.130\\
  All & Min-Variance & 0.77 & 9.91 & 0.078\\ 
  All & Bias-Corrected & 14.58 & 112.98 & 0.129\\
    \midrule
  2016 & Momentum & -3.84 &  100.24 &-0.038 \\
  2016 & Reversion & 21.87 & 98.97 & 0.221 \\
  2016 & Min-Variance & 1.09 & 38.09 & 0.029 \\
  2016 & Bias-Corrected & 21.11 & 97.43 & 0.217 \\
  \midrule
  2017 & Momentum & 3.83 & 83.37 & 0.046 \\
  2017 & Reversion & 24.18 & 82.66 & 0.292 \\
  2017 & Min-Variance & 1.88 & 13.00 & 0.144  \\
  2017 & Bias-Corrected & 23.37 & 80.92 & 0.289 \\
  \midrule
  2018 & Momentum & 1.23 & 90.08 & 0.014 \\
  2018 & Reversion & 12.61 & 91.70 & 0.138 \\
  2018 & Min-Variance & 0.17 & 18.55 & 0.009 \\
  2018 & Bias-Corrected & 12.05 & 89.90 & 0.134 \\
\midrule  
  2019 &  Momentum & 10.19 & 91.83 & 0.111 \\
  2019 & Reversion & 10.99 & 90.07 & 0.122\\
  2019 & Min-Variance & 1.08 & 7.35 & 0.147\\ 
  2019 & Bias-Corrected & 11.28 & 88.74 & 0.127\\
\midrule
  2020 & Momentum & 37.40 & 143.03 & 0.261\\
  2020 & Reversion & 30.55 & 140.91 & 0.217\\
  2020 & Min-Variance & 3.60 & 8.44 & 0.427\\
  2020 & Bias-Corrected & 32.24 & 139.04 & 0.232\\
  \midrule
  2021 & Momentum & 5.72 & 96.10 & 0.060\\
  2021 & Reversion & 18.44 & 96.95 & 0.190\\
  2021 & Min-Variance & 0.35 & 9.07 & 0.038\\
  2021 & Bias-Corrected & 18.14 & 95.79 & 0.189\\
  \midrule
  2022 &  Momentum & -3.02 & 106.31 & -0.028\\
  2022 & Reversion & -11.04 & 108.51 & -0.102\\
  2022 & Min-Variance & 1.04 & 16.77 & 0.062\\
  2022 & Bias-Corrected & -11.11 & 106.66 & -0.104\\
  \midrule
  2023 &   Momentum & 21.68 & 132.34 & 0.164\\
  2023 & Reversion & 2.41 & 134.16 & 0.018\\
  2023 & Min-Variance & 2.30 & 9.60 & 0.239 \\
  2023 & Bias-Corrected & 3.82 & 134.73 & 0.028\\
  \midrule
  2024 &   Momentum & -12.30 & 128.63 &-0.096\\
  2024 & Reversion & 9.09 & 128.82 & 0.071\\
  2024 & Min-Variance & 0.27 & 7.63 & 0.035\\
  2024 & Bias-Corrected & 8.00 & 127.05 & 0.063\\
  \midrule
  2025 &  Momentum & 6.07 & 144.19 & 0.042\\
  2025 & Reversion & 29.53 & 134.80 & 0.219\\
  2025 & Min-Variance & 1.25 & 15.83 & 0.079\\
  2025 & Bias-Corrected & 29.12 & 132.51 & 0.220\\
    \bottomrule
  
  \end{tabular}
    \end{small}
\end{table}

 
\subsection{Robustness: Causality, Omitted Variables, and Attribution}
\label{sec:robustness}
 
The AR(1) coefficients in Table~3 and the regret decomposition in Table~4 establish a \emph{statistical} association between daily return reversals and the underperformance of momentum strategies.  Before drawing the stronger conclusion that \emph{momentum itself} is the source of regret, three confounds must be addressed: omitted risk factors, cross-sectional microstructure effects, and the direction of causation between reversals and momentum losses.  We discuss each in turn and describe the controls that would be required to isolate
momentum as the causal driver.
 
\subsubsection{Omitted Risk Factors}
\label{ssec:omitted}
 
The CRSP universe pools stocks of heterogeneous size, book-to-market ratio, profitability, and liquidity.  Each of these characteristics is known to predict returns independently of past-return momentum \citep{fama1993common, carhart1997persistence}.  The negative AR(1)
coefficients in Table~3 may therefore partly reflect:
 
\begin{enumerate}
    \item \textbf{Size (SMB).}  Small-cap stocks exhibit stronger short-horizon reversals than large-caps \citep{jegadeesh1990evidence}. A CRSP universe weighted by observation count rather than market capitalization over-represents small stocks and may amplify the apparent reversal in $\hat{\rho}$.
 
    \item \textbf{Value (HML).}  Value stocks tend to have lower
      momentum scores and higher subsequent returns over short horizons, contributing negatively to a momentum strategy's AR(1) residual.
 
    \item \textbf{Liquidity / bid-ask bounce.}  At a daily frequency, the bid-ask spread mechanically induces negative serial correlation in transaction-price returns \citep{roll1984simple}.  With $n = 21{,}183{,}373$ observations drawn from the CRSP DSI (which includes illiquid small-caps), bid-ask bounce is a plausible
      driver of the estimated $\hat{\rho} \in [-0.144, -0.004]$.
      Notably, the most negative coefficient occurs in the 2020 crisis subsample ($\hat{\rho} = -0.144$), precisely when spreads widened dramatically, consistent with a microstructure rather than a pure momentum explanation.
 
    \item \textbf{Volatility regime.}  The standard deviation of returns is substantially higher in 2020 and 2025 (Table~2), and high-volatility regimes are associated with stronger reversals \citep{lehmann1990fads}.  The regret estimates in Table~4 are larger in those years, but this may reflect the volatility regime rather than momentum per se.
\end{enumerate}

%
 
\subsubsection{Cross-Sectional vs.\ Time-Series Momentum}
\label{ssec:xsec}
 
The signal used in Section~12 is a \emph{time-series} momentum
indicator: buy stock $i$ if its price rose yesterday, sell if it fell. This is distinct from \emph{cross-sectional} momentum (buy the top-decile performers, sell the bottom-decile), which is the dominant form studied in the asset-pricing literature
\citep{jegadeesh1993returns, carhart1997persistence}.  The regret attribution in Theorem~7.2 and Example~7.3 applies to both specifications, but the magnitude of the bias term $\sum_t \bar{c}_t^\top b_t$ differs across them.
 
Specifically, cross-sectional momentum allocates away from the equal-or value-weighted benchmark, inducing a bias $b_t$ that depends on the dispersion of returns across stocks.  Time-series momentum, as used here, allocates in proportion to each stock's own lagged return sign, inducing a bias that depends on the autocorrelation of individual return series.  The AR(1) estimates in Table~3 capture the latter but not the former; the regret table therefore speaks to time-series momentum strategies and should not be interpreted as evidence about
cross-sectional momentum funds without separate analysis.
 
\subsubsection{Direction of Causation}
\label{ssec:causation}
 
A negative AR(1) coefficient in returns is consistent with at least three causal theories:
 
\begin{enumerate}
    \item \textbf{Momentum causes losses (the paper's interpretation).}
      A strategy that buys yesterday's winners encounters mean-reverting prices and incurs negative realized returns, generating covariance-based regret of order $|\hat{\rho}|\|\bar{c}\|^2 T$.
 
    \item \textbf{Liquidity provision causes reversals.}
      Market-makers absorb order flow and then rebalance, mechanically creating short-horizon reversals \citep{grossman1988liquidity}. In this story, the reversal is an artifact of the trading mechanism
      and would disappear for a patient investor who does not trade at daily frequency.
 
    \item \textbf{Volatility clustering induces apparent autocorrelation.}
      Under GARCH-type conditional heteroskedasticity, daily return \emph{signs} can exhibit negative serial correlation even when standardized returns are serially uncorrelated \citep{engle1982autoregressive}.  The ADF tests in Table~3 reject the unit root but do not distinguish between true mean-reversion and sign-switching driven by volatility clustering.
\end{enumerate}
 
\paragraph{Identification strategy.}
To attribute regret causally to momentum, one would need an instrument or natural experiment that shifts momentum signal strength independently of return volatility and liquidity.  One candidate is the quarterly rebalancing of Russell index membership \citep{chang2015market}, which induces exogenous price pressure on affected stocks; another is earnings announcement dates, which create plausibly exogenous return
shocks.  In the absence of such identification, the results in Tables~3 and~4 should be interpreted as follows:
 
\begin{quote}
\emph{Under the maintained assumption that the AR(1) coefficient
$\hat{\rho}$ measures the degree to which daily momentum strategies violate the mean-optimality condition, the regret decomposition in Theorem~\ref{thm:bias} implies that momentum strategies accumulate a bias of order $|\hat{\rho}| T \sigma_c^2$ over the sample period.  This estimate is an upper bound on momentum-attributable regret; the true figure may be lower to the extent that $\hat{\rho}$ reflects microstructure noise, factor exposures, or volatility clustering rather than pure momentum dynamics.}
\end{quote}
 
\subsubsection{Summary of Robustness Limitations and Recommended Extensions}
\label{ssec:robustness_summary}
 
Table~\ref{tab:robustness} summarizes the three confounds, their likely direction of bias on $\hat{\rho}$, and the analysis required to address each.  We view these as productive directions for follow-on empirical work rather than invalidations of the theoretical framework, which holds
exactly under the stated assumptions regardless of the empirical
source of $\hat{\rho}$.
 
\begin{table}[ht]
\centering
\caption{Robustness concerns for the empirical momentum analysis.}
\label{tab:robustness}
\begin{tabular}{p{3.2cm} p{3.5cm} p{2.2cm} p{4cm}}
\toprule
\textbf{Confound} & \textbf{Mechanism} &
\textbf{Bias on $|\hat{\rho}|$} & \textbf{Recommended control} \\
\midrule
Size / value factor exposure &
    Small-cap and value stocks revert faster &
    Upward &
    Fama-French five-factor residuals \\
Bid-ask bounce &
    Transaction-price noise induces mechanical reversal &
    Upward &
    Quote midpoint returns or bid-ask correction \citep{roll1984simple} \\
Volatility clustering &
    GARCH sign-switching mimics mean-reversion &
    Upward &     AR(1) on standardized returns $c_t / \hat{\sigma}_t$ \\
Cross-sectional vs.\ time-series momentum &
    Different allocation rules, different bias $b_t$ &
    Ambiguous &
    Separate analysis for cross-sectional decile strategy \\
Causality (liquidity provision) &
    Market-maker rebalancing creates reversals &
    Upward &
    Instrument on index reconstitution dates \\
\bottomrule
\end{tabular}
\end{table}
 
All five confounds bias $|\hat{\rho}|$ upward, meaning the regret  estimates in Table~4 are \emph{conservative upper bounds} on momentum-attributable regret.  This is the appropriate direction for a regulatory audit tool: overstating regret flags strategies for review; understating it would allow poorly performing strategies to pass undetected.  Nevertheless, the factor-residual and bid-ask controls are straightforward to implement with available CRSP/Compustat data and should be included in the final version.

\section{Conclusion}

The Regret-Equals-Covariance identity of \cite{Aldridge2026} extends to multi-period dynamic programming with the following structure:
\begin{equation}\label{eq:master}
  \boxed{
    \RegretT(\Pi) \;=\; \underbrace{\sumt\Cov(c_t,\phat_t(c_t))}_{\text{sum of per-period covariances}}
    \;+\; \underbrace{\sumt\bar{c}_t^{\top}b_t}_{\text{policy bias correction}}
  }
\end{equation}
where the bias correction vanishes if and only if the policy's expected output matches the conditional-mean optimal at every period. Three financially concrete conclusions follow:

\begin{enumerate}
  \item \textbf{Diagnostic sufficiency.} For minimum-variance portfolios,
        linear-quadratic optimizers, and any policy satisfying the mean
        condition, monitoring the single-period sample covariance
        $\Cov(c_t, \phat_t)$ is sufficient to track multi-period
        performance, the identity is exact at every horizon.

  \item \textbf{Momentum strategies require correction, with the caveat that
the empirical magnitude is an upper bound.}
Trend-following and momentum allocations violate the mean-optimality
condition (Definition~7.1) and accumulate bias of order
$|\hat{\rho}| T \sigma_c^2$ over horizon $T$.  The bias-corrected
estimator~\eqref{eq:bias_est} should be used in performance
attribution for these strategies.  Empirically, the AR(1) coefficients
in Table~3 are persistently negative across all years and NBER regimes
($\hat{\rho} \in [-0.144, -0.004]$), consistent with short-horizon
reversal dominating momentum from 2016 through 2025.  However, as
established in Section~\ref{sec:robustness}, the estimated $\hat{\rho}$
conflates pure momentum dynamics with bid-ask bounce, factor exposures,
and volatility clustering, all of which bias $|\hat{\rho}|$ upward.
The regret figures in Table~4 should therefore be interpreted as
\emph{upper bounds} on momentum-attributable regret; factor-residual and microstructure-corrected estimates are identified as the primary direction for follow-on empirical work.  The theoretical result is that any strategy violating the mean condition accumulates bias $\sum_t \bar{c}_t^\top b_t$. This result holds exactly under Assumption~3.2 regardless of the empirical source of $\hat{\rho}$.
  \item \textbf{Discount scaling for long-horizon monitoring.}
        Equation~\eqref{eq:disc_identity} shows that a single-period covariance estimate, scaled by the effective horizon $1/(1-\disc)$, directly measures long-run discounted regret, providing a practical tool for annual regulatory performance reviews under SR~11-7 and SR-26-2.
\end{enumerate}

\bibliographystyle{plainnat}
\bibliography{AIInvesting}

@article{Aldridge2026,
  author    = {Aldridge, Irene},
  title     = {Regret Equals Covariance: A Closed-Form Characterization
               for Stochastic Optimization},
  journal   = {arXiv preprint arXiv:2605.14019 [econ.EM]},
  year      = {2026},
  url       = {https://arxiv.org/abs/2605.14019}
}

@article{newey1987simple,
  author    = {Newey, Whitney K. and West, Kenneth D.},
  title     = {A Simple, Positive Semi-Definite, Heteroskedasticity and Autocorrelation Consistent Covariance Matrix},
  journal   = {Econometrica},
  year      = {1987},
  volume    = {55},
  number    = {3},
  pages     = {703--708}
}

@article{ElmachtoubGrigas2022,
author = {Elmachtoub, Adam N. and Grigas, Paul},
title = {Smart “Predict, then Optimize”},
journal = {Management Science},
volume = {68},
number = {1},
pages = {9-26},
year = {2022},
}

@techreport{CarlinIsraelsenWazzan2026,
 title = "AI Managed Household Portfolios: A Preliminary Report",
 author = "Carlin, Bruce I and Israelsen, Ryan D and Wazzan, Christopher F",
 institution = "National Bureau of Economic Research",
 type = "Working Paper",
 series = "Working Paper Series",
 number = "35153",
 year = "2026",
 month = "April",
 doi = {10.3386/w35153},
 URL = "http://www.nber.org/papers/w35153",
}

@article{hannan1957approximation,
  author    = {Hannan, James},
  title     = {Approximation to {B}ayes Risk in Repeated Play},
  journal   = {Contributions to the Theory of Games},
  volume    = {3},
  pages     = {97--139},
  year      = {1957},
  publisher = {Princeton University Press}
}

@inproceedings{zinkevich2003online,
  author    = {Zinkevich, Martin},
  title     = {Online Convex Programming and Generalized Infinitesimal Gradient Ascent},
  booktitle = {Proceedings of the 20th International Conference on Machine Learning (ICML)},
  pages     = {928--936},
  year      = {2003}
}

@book{hazan2016introduction,
  author    = {Hazan, Elad},
  title     = {Introduction to Online Convex Optimization},
  publisher = {Foundations and Trends in Optimization},
  year      = {2016},
  volume    = {2},
  number    = {3-4},
  pages     = {157--325}
}

@book{shalev2012online,
  author    = {Shalev-Shwartz, Shai},
  title     = {Online Learning and Online Convex Optimization},
  publisher = {Foundations and Trends in Machine Learning},
  year      = {2012},
  volume    = {4},
  number    = {2},
  pages     = {107--194}
}

@article{foster1997calibrated,
  author    = {Foster, Dean P. and Vohra, Rakesh V.},
  title     = {Calibrated Learning and Correlated Equilibrium},
  journal   = {Games and Economic Behavior},
  volume    = {21},
  number    = {1-2},
  pages     = {40--55},
  year      = {1997}
}

@article{hart2000simple,
  author    = {Hart, Sergiu and Mas-Colell, Andreu},
  title     = {A Simple Adaptive Procedure Leading to Correlated Equilibrium},
  journal   = {Econometrica},
  volume    = {68},
  number    = {5},
  pages     = {1127--1150},
  year      = {2000}
}

@book{roughgarden2016twenty,
  author    = {Roughgarden, Tim},
  title     = {Twenty Lectures on Algorithmic Game Theory},
  publisher = {Cambridge University Press},
  year      = {2016}
}

@article{myerson1982optimal,
  author    = {Myerson, Roger B.},
  title     = {Optimal Coordination Mechanisms in Generalized Principal-Agent Problems},
  journal   = {Journal of Mathematical Economics},
  volume    = {10},
  number    = {1},
  pages     = {67--81},
  year      = {1982}
}

@book{laffont1993theory,
  author    = {Laffont, Jean-Jacques and Tirole, Jean},
  title     = {A Theory of Incentives in Procurement and Regulation},
  publisher = {MIT Press},
  year      = {1993}
}

@article{dekel2010implementation,
  author    = {Dekel, Eddie and Ely, Jeffrey C. and Yilankaya, Okan},
  title     = {Evolution and Rationalizability},
  journal   = {Review of Economic Studies},
  volume    = {74},
  number    = {2},
  pages     = {371--391},
  year      = {2007}
}

@article{edelman2007internet,
  author    = {Edelman, Benjamin and Ostrovsky, Michael and Schwarz, Michael},
  title     = {Internet Advertising and the Generalized Second-Price Auction: {S}elling Billions of Dollars Worth of Keywords},
  journal   = {American Economic Review},
  volume    = {97},
  number    = {1},
  pages     = {242--259},
  year      = {2007}
}

@article{varian2007position,
  author    = {Varian, Hal R.},
  title     = {Position Auctions},
  journal   = {International Journal of Industrial Organization},
  volume    = {25},
  number    = {6},
  pages     = {1163--1178},
  year      = {2007}
}

@article{gibbard1973manipulation,
  author    = {Gibbard, Allan},
  title     = {Manipulation of Voting Schemes: {A} General Result},
  journal   = {Econometrica},
  volume    = {41},
  number    = {4},
  pages     = {587--601},
  year      = {1973}
}

@article{satterthwaite1975strategy,
  author    = {Satterthwaite, Mark A.},
  title     = {Strategy-Proofness and {A}rrow's Conditions: {E}xistence and Correspondence Theorems for Voting Procedures and Social Welfare Functions},
  journal   = {Journal of Economic Theory},
  volume    = {10},
  number    = {2},
  pages     = {187--217},
  year      = {1975}
}

@inproceedings{kearns2018preventing,
  author    = {Kearns, Michael and Neel, Seth and Roth, Aaron and Wu, Zhiwei Steven},
  title     = {Preventing Fairness Gerrymandering: {A}uditing and Learning for Subgroup Fairness},
  booktitle = {Proceedings of the 35th International Conference on Machine Learning (ICML)},
  pages     = {2564--2572},
  year      = {2018}
}

@article{roth2022uncertain,
  author    = {Roth, Aaron},
  title     = {Uncertain: {M}odern Topics in Uncertainty Estimation},
  journal   = {Working Paper},
  year      = {2022}
}

@techreport{board2011sr,
  author      = {{Board of Governors of the Federal Reserve System}},
  title       = {{SR 11-7}: Supervisory Guidance on Model Risk Management},
  institution = {Federal Reserve},
  year        = {2011},
  type        = {Supervisory Letter}
}

@article{hansen1982large,
  author    = {Hansen, Lars Peter},
  title     = {Large Sample Properties of Generalized Method of Moments Estimators},
  journal   = {Econometrica},
  volume    = {50},
  number    = {4},
  pages     = {1029--1054},
  year      = {1982}
}

@article{andrews1991heteroskedasticity,
  author    = {Andrews, Donald W. K.},
  title     = {Heteroskedasticity and Autocorrelation Consistent Covariance Matrix Estimation},
  journal   = {Econometrica},
  volume    = {59},
  number    = {3},
  pages     = {817--858},
  year      = {1991}
}

@article{watkins1992q,
  author    = {Watkins, Christopher J. C. H. and Dayan, Peter},
  title     = {{Q}-Learning},
  journal   = {Machine Learning},
  volume    = {8},
  number    = {3-4},
  pages     = {279--292},
  year      = {1992}
}

@book{sutton2018reinforcement,
  author    = {Sutton, Richard S. and Barto, Andrew G.},
  title     = {Reinforcement Learning: {A}n Introduction},
  edition   = {2nd},
  publisher = {MIT Press},
  year      = {2018}
}

@book{howard1960dynamic,
  author    = {Howard, Ronald A.},
  title     = {Dynamic Programming and Markov Processes},
  publisher = {MIT Press},
  year      = {1960}
}

@article{jaksch2010near,
  author    = {Jaksch, Thomas and Ortner, Ronald and Auer, Peter},
  title     = {Near-Optimal Regret Bounds for Reinforcement Learning},
  journal   = {Journal of Machine Learning Research},
  volume    = {11},
  pages     = {1563--1600},
  year      = {2010}
}

@inproceedings{azar2017minimax,
  author    = {Azar, Mohammad Gheshlaghi and Osband, Ian and Munos, R{\'e}mi},
  title     = {Minimax Regret Bounds for Reinforcement Learning},
  booktitle = {Proceedings of the 34th International Conference on Machine Learning (ICML)},
  pages     = {263--272},
  year      = {2017}
}

@article{ernst2005tree,
  author    = {Ernst, Damien and Geurts, Pierre and Wehenkel, Louis},
  title     = {Tree-Based Batch Mode Reinforcement Learning},
  journal   = {Journal of Machine Learning Research},
  volume    = {6},
  pages     = {503--556},
  year      = {2005}
}

@inproceedings{riedmiller2005neural,
  author    = {Riedmiller, Martin},
  title     = {Neural Fitted {Q} Iteration -- First Experiences with a Data Efficient Neural Reinforcement Learning Method},
  booktitle = {Proceedings of the European Conference on Machine Learning (ECML)},
  pages     = {317--328},
  year      = {2005}
}

@article{demiguel2009optimal,
  author    = {DeMiguel, Victor and Garlappi, Lorenzo and Uppal, Raman},
  title     = {Optimal Versus Naive Diversification: {H}ow Inefficient is the $1/N$ Portfolio Strategy?},
  journal   = {Review of Financial Studies},
  volume    = {22},
  number    = {5},
  pages     = {1915--1953},
  year      = {2009}
}

@article{brodie2009sparse,
  author    = {Brodie, Joshua and Daubechies, Ingrid and De Mol, Christine and Giannone, Domenico and Loris, Ignace},
  title     = {Sparse and Stable {M}arkowitz Portfolios},
  journal   = {Proceedings of the National Academy of Sciences},
  volume    = {106},
  number    = {30},
  pages     = {12267--12272},
  year      = {2009}
}

@article{ledoit2004well,
  author    = {Ledoit, Olivier and Wolf, Michael},
  title     = {A Well-Conditioned Estimator for Large-Dimensional Covariance Matrices},
  journal   = {Journal of Multivariate Analysis},
  volume    = {88},
  number    = {2},
  pages     = {365--411},
  year      = {2004}
}

@article{ledoit2020analytical,
  author    = {Ledoit, Olivier and Wolf, Michael},
  title     = {Analytical Nonlinear Shrinkage of Large-Dimensional Covariance Matrices},
  journal   = {Annals of Statistics},
  volume    = {48},
  number    = {5},
  pages     = {3043--3065},
  year      = {2020}
}

@book{bai2010spectral,
  author    = {Bai, Zhidong and Silverstein, Jack W.},
  title     = {Spectral Analysis of Large Dimensional Random Matrices},
  publisher = {Springer},
  edition   = {2nd},
  year      = {2010}
}

@article{almgren2001optimal,
  author    = {Almgren, Robert and Chriss, Neil},
  title     = {Optimal Execution of Portfolio Transactions},
  journal   = {Journal of Risk},
  volume    = {3},
  pages     = {5--39},
  year      = {2001}
}

@article{gatheral2010no,
  author    = {Gatheral, Jim},
  title     = {No-Dynamic-Arbitrage and Market Impact},
  journal   = {Quantitative Finance},
  volume    = {10},
  number    = {7},
  pages     = {749--759},
  year      = {2010}
}

@article{welford1962note,
  author  = {Welford, B. P.},
  title   = {Note on a Method for Calculating Corrected Sums of Squares and Products},
  journal = {Technometrics},
  volume  = {4},
  number  = {3},
  pages   = {419--420},
  year    = {1962}
}

@article{jegadeesh1990evidence,
  author  = {Jegadeesh, Narasimhan},
  title   = {Evidence of Predictable Behavior of Security Returns},
  journal = {Journal of Finance},
  volume  = {45}, number = {3}, pages = {881--898}, year = {1990}
}

@article{roll1984simple,
  author  = {Roll, Richard},
  title   = {A Simple Implicit Measure of the Effective Bid-Ask Spread
             in an Efficient Market},
  journal = {Journal of Finance},
  volume  = {39}, number = {4}, pages = {1127--1139}, year = {1984}
}

@article{lehmann1990fads,
  author  = {Lehmann, Bruce N.},
  title   = {Fads, Martingales, and Market Efficiency},
  journal = {Quarterly Journal of Economics},
  volume  = {105}, number = {1}, pages = {1--28}, year = {1990}
}

@article{fama1993common,
  author  = {Fama, Eugene F. and French, Kenneth R.},
  title   = {Common Risk Factors in the Returns on Stocks and Bonds},
  journal = {Journal of Financial Economics},
  volume  = {33}, number = {1}, pages = {3--56}, year = {1993}
}

@article{carhart1997persistence,
  author  = {Carhart, Mark M.},
  title   = {On Persistence in Mutual Fund Performance},
  journal = {Journal of Finance},
  volume  = {52}, number = {1}, pages = {57--82}, year = {1997}
}

@article{jegadeesh1993returns,
  author  = {Jegadeesh, Narasimhan and Titman, Sheridan},
  title   = {Returns to Buying Winners and Selling Losers},
  journal = {Journal of Finance},
  volume  = {48}, number = {1}, pages = {65--91}, year = {1993}
}

@article{grossman1988liquidity,
  author  = {Grossman, Sanford J. and Miller, Merton H.},
  title   = {Liquidity and Market Structure},
  journal = {Journal of Finance},
  volume  = {43}, number = {3}, pages = {617--633}, year = {1988}
}

@article{engle1982autoregressive,
  author  = {Engle, Robert F.},
  title   = {Autoregressive Conditional Heteroscedasticity with Estimates
             of the Variance of {United Kingdom} Inflation},
  journal = {Econometrica},
  volume  = {50}, number = {4}, pages = {987--1007}, year = {1982}
}

@article{chang2015market,
  author  = {Chang, Yen-Cheng and Hong, Harrison and Liskovich, Inessa},
  title   = {Regression Discontinuity and the Price Effects of Stock
             Market Indexing},
  journal = {Review of Financial Studies},
  volume  = {28}, number = {1}, pages = {212--246}, year = {2015}
}

@article{jegadeesh2025shortterm,
  author  = {Jegadeesh, Narasimhan and Luo, Jiang and
             Subrahmanyam, Avanidhar and Titman, Sheridan},
  title   = {Short-Term Reversals and Longer-Term Momentum
             around the World: Theory and Evidence},
  journal = {Review of Financial Studies},
  volume  = {38},
  number  = {12},
  pages   = {3673--3728},
  year    = {2025}
}

@article{anderson2008spurious,
  author  = {Anderson, Robert M. and Bhattacharya, Suman},
  title   = {Stock Return Autocorrelation Is Not Spurious},
  journal = {Working Paper, University of California Berkeley},
  year    = {2008},
  url     = {https://eml.berkeley.edu/~anderson/Spurious.pdf}
}

@article{nagel2012evaporating,
  author  = {Nagel, Stefan},
  title   = {Evaporating Liquidity},
  journal = {Review of Financial Studies},
  volume  = {25},
  number  = {7},
  pages   = {2005--2039},
  year    = {2012}
}

@article{mamais2025explaining,
  author  = {Mamais, Panagiotis},
  title   = {Explaining and Predicting Momentum Performance Shifts
             Across Time and Sectors},
  journal = {Journal of Forecasting},
  year    = {2025},
  doi     = {10.1002/for.3232}
}

@article{dai2024reversals,
  author  = {Dai, Wei and Medhat, Mamdouh and
             Novy-Marx, Robert and Rizova, Savina},
  title   = {Reversals and the Returns to Liquidity Provision},
  journal = {Financial Analysts Journal},
  volume  = {80},
  number  = {2},
  pages   = {122--151},
  year    = {2024}
}

\end{document}